\pdfoutput=1

\documentclass[review,12pt]{article}
\usepackage{graphicx}
\usepackage{epstopdf}
\usepackage{caption}
\usepackage{subfig}
\usepackage[utf8]{inputenc}

\usepackage{amssymb}
\usepackage{array}
\newcolumntype{C}[1]{>{\centering\let\newline\\\arraybackslash\hspace{0pt}}m{#1}}
\usepackage{multirow}
\usepackage{amsthm}
\usepackage{amsmath}
\newtheorem{theor}{Theorem}
\usepackage[ruled]{algorithm2e}
\usepackage{hyperref}
\usepackage{apacite}
\usepackage[pscoord]{eso-pic}% The zero point of the coordinate systemis the lower left corner of the page (the default).
\usepackage{textcomp}

\newcommand{\placetextbox}[3]{% \placetextbox{<horizontal pos>}{<vertical pos>}{<stuff>}
  \setbox0=\hbox{#3}% Put <stuff> in a box
  \AddToShipoutPictureFG*{% Add <stuff> to current page foreground
    \put(\LenToUnit{#1\paperwidth},\LenToUnit{#2\paperheight}){\vtop{{\null}\makebox[0pt][c]{#3}}}%
  }%
}%

\begin{document}

\placetextbox{0.5}{0.97}{\normalfont \textcopyright 2020. This manuscript version is made available under the CC-BY-NC-ND 4.0 license}

\placetextbox{0.5}{0.955}{\normalfont \url{http://creativecommons.org/licenses/by-nc-nd/4.0/}}

\placetextbox{0.5}{0.93}{\normalfont Author accepted manuscript, published in}

\placetextbox{0.5}{0.915}{\normalfont ``Expert Systems With Applications, 2019, 122: 262--280''.}

\placetextbox{0.5}{0.90}{\normalfont DOI: \url{https://doi.org/10.1016/j.eswa.2019.01.008}.}%

%\begin{frontmatter}



%% \title{Title\tnoteref{ICA}}

\begin{center}
{\Large Application of independent component analysis and TOPSIS to deal with dependent criteria in multicriteria decision problems}
\end{center}


\vspace{1cm}

Guilherme Dean Pelegrina$^{a*}$, Leonardo Tomazeli Duarte$^b$, João Marcos Travassos Romano$^a$

\vspace{0.5cm}

{\small $^a$School of Electrical and Computer Engineering (FEEC), University of Campinas (UNICAMP), Campinas, Brazil; guipelegrina@gmail.com, romano@dmo.fee.unicamp.br

$^b$School of Applied Sciences (FCA), University of Campinas (UNICAMP), Limeira, Brazil; leonardo.duarte@fca.unicamp.br

$^*$Corresponding author: Phone number: +33 06 58 94 09 93}

\vspace{1cm}

\noindent \textbf{Abstract}

\noindent A vast number of multicriteria decision making methods have been developed to deal with the problem of ranking a set of alternatives evaluated in a multicriteria fashion. Very often, these methods assume that the evaluation among criteria is statistically independent. However, in actual problems, the observed data may comprise dependent criteria, which, among other problems, may result in biased rankings. In order to deal with this issue, we propose a novel approach whose aim is to estimate, from the observed data, a set of independent latent criteria, which can be seen as an alternative representation of the original decision matrix. A central element of our approach is to formulate the decision problem as a blind source separation problem, which allows us to apply independent component analysis techniques to estimate the latent criteria. Moreover, we consider TOPSIS-based approaches to obtain the ranking of alternatives from the latent criteria. Results in both synthetic and actual data attest the relevance of the proposed approach.

\vspace{0.5cm}






\newpage

\section{Introduction}
\label{sec:intro}

Either in personal life or in industrial environments, most of the decisions involve situations where there is a large amount of information. In that respect, the development of appropriate methods which can deal with those situations has become an important issue in various research domains, such as operational research, management science, economics and computer science~\cite{Zavadskas2014}.

A typical problem addressed in decision making is how to obtain a ranking of a set of alternatives given their evaluation according to a set of criteria. In order to deal with this problem, several Multicriteria Decision Making (MCDM) methods have been developed along the last decades~\cite{Figueira2005,Tzeng2011}. Examples of existing approaches include the Simple Additive Weighting (SAW) method~\cite{Tzeng2011}, the Analytic Hierarchy Process (AHP)~\cite{Saaty1987}, the ELECTRE methods~\cite{Roy1968,Govindan2016} and the Technique for Order Preference by Similarity to Ideal Solution (TOPSIS)~\cite{Hwang1981}. For instance, the TOPSIS and some extended versions have been used in several applications~\cite{Behzadian2012,Zavadskas2016,Yoon2017}, such as in energy planning~\cite{Kaya2011} and factory maintenance~\cite{Cables2012} problems, lean performance~\cite{Kumar2013} and pure-play e-commerce companies~\cite{Bai2014} evaluation, and identification of spreading ability of nodes~\cite{Hu2016}.

However, the existing approaches often consider that the evaluation among the observed criteria are independent\footnote{In this work, independence (or dependence) among criteria will refer to statistical independence (or statistical dependence) among the observed evaluations with respect to the decision criteria. See~\cite{Roy1996} for an interesting discussion on this topic.}, which may not be true in real applications. As a consequence, the results may be biased toward alternatives that have good evaluations in two or more dependent criteria, since they measure similar characteristics. One may cite as an example the problem of ranking a set of students according to their grades in a set of subjects~\cite{Grabisch1996}, such as human resources, physics and linear algebra. In this scenario, since both physics and linear algebra measure similar competences, the ranking may be biased towards students with good grades in these subjects. 

In the literature, one may find several approaches that aim at dealing with dependent criteria in MCDM problems. For instance, a well-known is the application of Choquet integrals~\cite{Grabisch1996,Grabisch2010}, which is a non-linear aggregator that can be used to model interactions (redundancy and/or synergy) among criteria through fuzzy parameters. Moreover, Figueira, Greco and Roy~\citeyear{Figueira2009} proposed an extended version of ELECTRE methods, which performs interaction among criteria by applying a new concordance index. Zhu et al.~\citeyear{Zhu2016} applied Principal Component Analysis (PCA) technique and TOPSIS approach to deal with correlation among the decision data. Wang~\citeyear{Wang2015} and Wang et al.~\citeyear{Wang2017} also combined TOPSIS with PCA and variable clustering, respectively, to address multicollinearity between criteria.

Another TOPSIS-based approach used to deal with dependent criteria in MCDM problems is the TOPSIS-M~\cite{Vega2014}. The original version of TOPSIS relies on a global evaluation of each alternative based on the Euclidean distances between the given alternative and both positive and negative ideal alternatives. On the other hand, TOPSIS-M addresses the dependent criteria by applying the Mahalanobis distance~\cite{Mahalanobis1936,DeMaesschalck2000} instead of the Euclidean one. This feature has motivated the application of TOPSIS-M in several problems. For instance, Antuchevi\v{c}iene et al.~\citeyear{Antuche2010} and Chang et al.~\citeyear{Chang2010} applied the TOPSIS-M approach to rank different types of areas for building redevelopment and to evaluate mutual funds, respectively. Wang and Wang~\citeyear{Wang2014} discussed the same procedure in the context of the Chinese high-tech industry. Chen and Lu~\citeyear{Chen2015} and Chen~\citeyear{Chen2017} used a fuzzy TOPSIS-M to obtain the scores of insurance companies and investment policies, respectively. Wang, Li and Zheng~\citeyear{Wang2018} considered an entropy weighted TOPSIS-M approach in China energy regulation.

In view of well-established theoretical grounds of the TOPSIS approach,  and motivated by the existence of practical problems that have been dealt with by means of TOPSIS-based methods that are able to deal with correlated criteria, our proposal shall be built upon the TOPSIS methodology. In particular, we shall consider in our proposals the classical TOPSIS and on TOPSIS-M variant. An important motivation of our work comes from the observation that, although the TOPSIS-M approach has been used to deal with dependent criteria in MCDM problems, in some situations, the covariance matrix used in Mahalanobis distance may not be sufficient to capture the complexity behind the interactions between the observed criteria. As will be discussed in this paper, the distances calculated in TOPSIS-M can be seen as Euclidean ones in a transformed space in which the decision data are not correlated anymore. In other words, this approach aims at finding an alternative representation of the data in which the distances calculation should be conducted. However, the decorrelation procedure of TOPSIS-M approach does not necessarily imply in a representation in which the decision data are independent~\cite{Papoulis2002}. Therefore, an approach that takes into account the independence among the decision data can better deal with dependent criteria in MCDM problems.

An interesting aspect in the above-mentioned scenario is that the procedure to find an alternative representation of the decision data can be seen as a task of recovering a set of independent latent (hidden) data from the observed criteria. This is a well-know problem in signal processing, called Blind Source Separation (BSS)~\cite{Comon2010}, whose aim is to retrieve a set of signal sources based on a mixture of these sources, without the knowledge of the mixing process. In that respect, our proposal relies on the formulation of the addressed MCDM problem as a BSS problem and apply an Independent Component Analysis (ICA) technique~\cite{Hyvarinen2001} in order to estimate the latent criteria, which are assumed to be mutually independent. Since the latent criteria can be seen as the alternative representation of the decision data in which the criteria are independent, one can use the estimates as inputs of TOPSIS methods. Therefore, this paper proposes two different approaches to deal with dependent criteria in MCDM problems. The first one, called ICA-TOPSIS, comprises the application of an ICA technique to estimate the latent criteria and, thus, perform the TOPSIS on this retrieved data. The other one, called ICA-TOPSIS-M, also applies an ICA technique to estimate the latent criteria, but uses this retrieved data only to determine the positive and negative ideal alternatives, which will be used as an input in a modified TOPSIS-M approach.

The next sections are organized as follows. In Section~\ref{sec:prob}, we present the MCDM problem addressed in this paper. Section~\ref{sec:theo} discusses the theoretical aspects of TOPSIS approaches and BSS methods. Both ICA-TOPSIS and ICA-TOPSIS-M approaches proposed in this paper are described in Section~\ref{sec:proposal}. In Section~\ref{sec:exper}, we present a set of numerical experiments and the obtained results. Finally, in Section~\ref{sec:concl}, we highlight our conclusions and future perspectives.

It is worth mentioning that the conference paper~\cite{Pelegrina2018} presented a preliminary discussion on the application of ICA methods to deal with dependent criteria in MCDM problems. However, in this paper, we provide a novel theoretic result with respect to the TOPSIS-M approach and improve the set of numerical experiments, which includes scenarios with more than two decision criteria and also with different number of alternatives. Moreover, we apply our proposal on real data.

\section{Problem statement}
\label{sec:prob}

The MCDM problem addressed in this paper consists in ranking a set of $K$ alternatives $\mathcal{A} = \left[ \mathcal{A}_1, \mathcal{A}_2, \ldots, \mathcal{A}_K \right]$ based on an observed decision data $\mathbf{V}$, defined by
\begin{equation}
\label{eq:dec_data}
\mathbf{V}=\left[\begin{array}{cccc}
v_{1,1} & v_{1,2} & \ldots & v_{1,K} \\
v_{2,1} & v_{2,2} & \ldots & v_{2,K} \\ 
\vdots & \vdots & \ddots & \vdots \\
v_{M,1} & v_{M,2} & \ldots & v_{M,K}
\end{array}
\right],
\end{equation}
where $v_{m,k}$ represents the evaluation of alternative $\mathcal{A}_k$ ($k=1, 2, \ldots, K$) with respect to the criterion $\mathcal{C}_m$ ($m=1, 2, \ldots, M$). In order to obtain the ranking of the alternatives, one derives a global evaluation for each $\mathcal{A}_k$ through an aggregation procedure on $v_{m,k}$, $j=1, \ldots, M$, and based on a set of weights $\mathbf{w}=[w_{1}, w_{2}, \ldots, w_{M}]$, which represent the degree of ``importance'' for the criterion $\mathcal{C}_m$ in the decision problem\footnote{Although it is not required in all MCDM methods, one generally defines $w_m \geq 0$, $m=1, 2, \ldots, M$, and $\sum_{m=1}^M{w_m}=1$.}.

Generally, a MCDM method is applied directly on the observed decision data $\mathbf{V}$. However, since we are interested in MCDM problems in which the decision data $\mathbf{V}$ can be seen as a set of dependent criteria, we consider that an alternative representation of the decision data, in which the criteria are independent, should be used as input of TOPSIS method. The procedure of extracting this new representation from the decision data can be seen as a problem of retrieving independent latent factors from a set of mixtures of these factors.

Mathematically, consider that the independent latent data is defined by 
\begin{equation}
\label{eq:lat_data}
\mathbf{L}=\left[\begin{array}{cccc}
l_{1,1} & l_{1,2} & \ldots & l_{1,K} \\
l_{2,1} & l_{2,2} & \ldots & l_{2,K} \\ 
\vdots & \vdots & \ddots & \vdots \\
l_{N,1} & l_{N,2} & \ldots & l_{N,K}
\end{array}
\right],
\end{equation}
where $l_{n,k}$ represents the evaluation of alternative $\mathcal{A}_k$ with respect to the independent latent criteria\footnote{We define the set of independent latent criteria as $\mathcal{L}=\left[ \mathcal{L}_1, \mathcal{L}_2, \ldots, \mathcal{L}_N \right]$.} $\mathcal{L}_n$. Therefore, by considering a signal processing formulation, data $\mathbf{V}$ can be modeled as
\begin{equation}
\label{eq:mix_prop}
\mathbf{V} = \mathbf{A}\mathbf{L} + \mathbf{G},
\end{equation}
where $\mathbf{A}=\left( a_{m,n} \right) \in \mathbb{R}^{M \times N}$ is the mixing matrix (which provides the linear mixture of the independent data) and $\mathbf{G}=\left( g_{m,k} \right) \in \mathbb{R}^{M \times K}$ is the additive white Gaussian noise (AWGN) matrix (which includes randomness in the mixing process). For instance, the evaluation of alternative $\mathcal{A}_k$ with respect to the observed criterion $\mathcal{C}_m$ comprises a linear combination of the evaluations of alternative $\mathcal{A}_k$ with respect to the set of independent latent criteria $\mathcal{L}$, i.e. $v_{m,k} = a_{m,1}l_{1,k} + a_{m,2}l_{2,k} + \ldots + a_{m,N}l_{N,k} + g_{m,k}$.

If we consider that $\mathbf{L}$ is the alternative representation of the decision data that should be used to derive the ranking of alternatives, the application of TOPSIS directly on the dependent\footnote{Since mixing matrix $\mathbf{A}$ provides a linear combination of the latent criteria $\mathbf{L}$, one can consider that the observed data $\mathbf{V}$ is composed by dependent criteria.} criteria $\mathbf{V}$ may lead to biased results. Therefore, the main concern in this problem is how to adjust a transformation matrix $\mathbf{B}$ that provides the estimates
\begin{equation}
\label{eq:demix_mcdm}
\mathbf{\hat{L}} = \mathbf{B}\mathbf{V}
\end{equation}
for the independent data $\mathbf{L}$, i.e. the alternative representation of the decision criteria in which the data are independent. Figure~\ref{fig:problem_repres} illustrates this scenario, in which the procedure used to adjust $\mathbf{B}$ can be seen as a preprocessing step. It is worth mentioning that, in this paper, we consider only to deal with a linear dependence among criteria. However, our approach can be adapted in order to consider nonlinear relations.

\begin{figure}[ht]
\centering
\includegraphics[height=2.85cm]{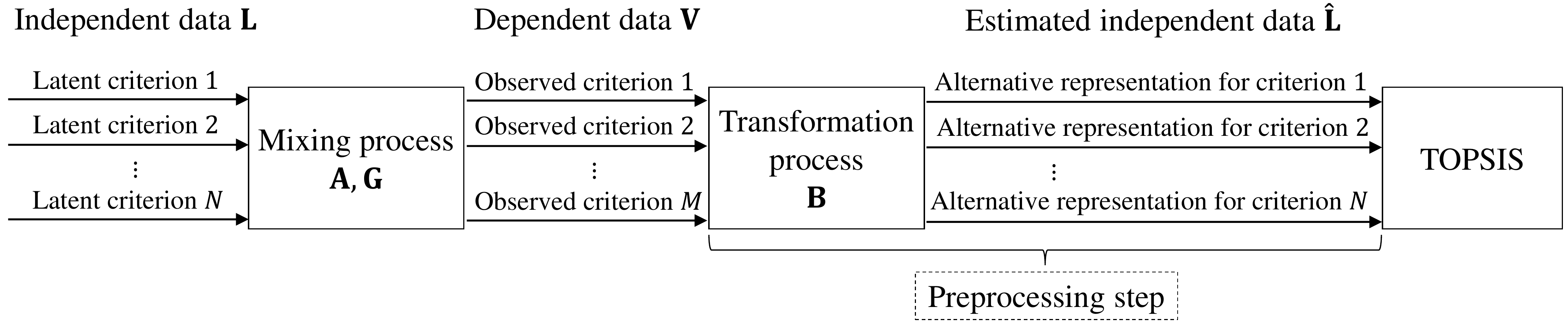}
\caption{Application of TOPSIS in an alternative representation of the decision data.}
\label{fig:problem_repres}
\end{figure}

Since we search for an alternative representation of the observed decision criteria $\mathbf{V}$, it is not straightforward to assign a meaning for this data. However, the rationale behind such model is that the observed criteria can be seen as linear combinations of independent (hidden) criteria. Figure~\ref{fig:problem_example} illustrates the example mentioned in Section~\ref{sec:intro}. In this context, one may consider that the grades (dependent decision data $\mathbf{V}$) may be represented as a set of estimated latent competences (independent data $\mathbf{L}$) in areas such as social sciences, natural sciences and mathematics. For instance, there may be a correlation between the grades of linear algebra and physics because both depend on competences related to mathematics. 

\begin{figure}[ht]
\centering
\includegraphics[height=2.5cm]{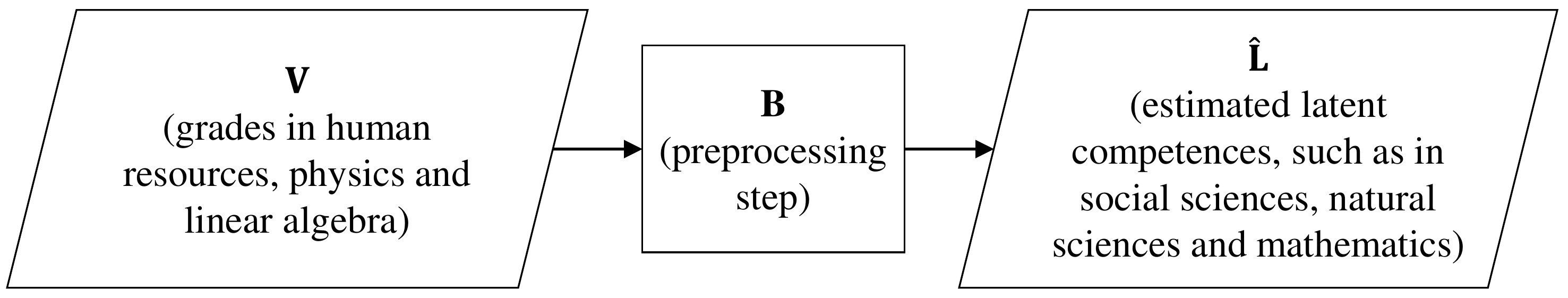}
\caption{Example of dependent decision criteria and estimated latent data.}
\label{fig:problem_example}
\end{figure}

\section{Theoretical background}
\label{sec:theo}

\subsection{TOPSIS and TOPSIS-M}
\label{sec:topsis}

The original version of TOPSIS calculates the global evaluation of each alternative (also called closeness measure) based on the Euclidean distances between the alternative and both positive and negative ideal alternatives. The main idea is to favour alternatives that are close to the positive ideal alternative and also far from the negative one. Algorithm~\ref{alg:topsis_e} describes the steps of this approach\footnote{We consider in this paper that all the criteria are to be maximized, i.e. the larger the better. However, if there are criteria to be minimized, some simple adaptations must be incorporated into algorithm steps. See~\cite{Hwang1981} for further details.}. One may note that $r_k \in [0, 1]$ close to $1$ indicates that the alternative $\mathcal{A}_k$ is close to the positive ideal alternative and far from the negative one. Therefore, the ranking is obtained according to the closeness measure in descending order. 
\begin{algorithm}
%\setstretch{1.00}
%\SetAlgoLined
    \caption{TOPSIS}
    \label{alg:topsis_e}
				\textbf{Input:} Decision data $\mathbf{V}$ and set of weights $\mathbf{w}$. \\
				\textbf{Output:} Closeness measure $\mathbf{r} = \left[r_1, r_2, \ldots, r_K \right]$. \\
				\textit{Step 1: Normalization}. For each evaluation $v_{m,k}$, perform the following normalization:
				\begin{equation*}
				u_{m,k} = \frac{v_{m,k}}{\sqrt{\sum_{k=1}^K{v_{m,k}^2}}}, \, \, \, m=1, \ldots, M, \, k=1, \ldots, K.
				\end{equation*}
				\textit{Step 2: Weighted normalization}. For each normalized evaluation $u_{m,k}$, perform the following weighted normalization:
				\begin{equation*}
				p_{m,k} = w_m u_{m,k}, \, \, \, m=1, \ldots, M, \, k=1, \ldots, K.
				\end{equation*}
				\textit{Step 3: Ideal alternatives determination}. Determine both positive and negative ideal alternatives (PIA and NIA, respectively), defined by
				\begin{equation*}
				PIA = \mathbf{p}^+ = \left\{p_{1}^+, p_{2}^+, \ldots, p_{M}^+\right\} \text{and } NIA = \mathbf{p}^- = \left\{p_{1}^-, p_{2}^-, \ldots, p_{M}^-\right\},
				\end{equation*}
				where $p_{m}^+=\max \{p_{m,k} | 1 \leq k \leq K \}$ and $p_{m}^-= \min \{p_{m,k} | 1 \leq k \leq K \}$, $m=1, \ldots, M$. \\
				\textit{Step 4: Euclidean distances}. Calculate the Euclidean distances from each alternative and both PIA and NIA:
				\begin{equation*}
				D_k^+ = \sqrt{\left(\mathbf{p}_k-\mathbf{p}^+\right)^T\left(\mathbf{p}_k-\mathbf{p}^+\right)}, \, \, \, k=1, \ldots, K,
				\end{equation*}
				and
				\begin{equation*}
				D_k^- = \sqrt{\left(\mathbf{p}_k-\mathbf{p}^-\right)^T\left(\mathbf{p}_k-\mathbf{p}^-\right)}, \, \, \, k=1, \ldots, K,
				\end{equation*}
				where $\mathbf{p}_k=[p_{1,k}, p_{2,k}, \ldots, p_{M,k}]$. \\
				\textit{Step 5: Closeness measure}. For each alternative, calculate the closeness measure:
				\begin{equation*}
				r_k = \frac{D_k^-}{D_k^+ + D_k^-}, \, \, \, k=1, \ldots, K.
				\end{equation*}
\end{algorithm}

The classical TOPSIS method is typically applied in MCDM problems in which the decision criteria are independent. However, the case of dependent decision criteria motivated the development of modified versions of the classical TOPSIS method. For instance, the TOPSIS-M method, which is described in Algorithm~\ref{alg:topsis_m}, addresses the dependence among criteria by exploiting second-order statistics (covariance matrix). An interesting aspect of the TOPSIS-M method is that, since the Mahalonobis distance considers the information about the covariance matrix of the observed decision data, this method addresses the issue of dependent criteria by performing a decorrelation procedure. This result can be understood under the light of the following theorem:

\begin{algorithm}
%\setstretch{1.00}
%\SetAlgoLined
    \caption{TOPSIS-M}
    \label{alg:topsis_m}
				\textbf{Input:} Decision data $\mathbf{V}$ and set of weights $\mathbf{w}$. \\
				\textbf{Output:} Closeness measure $\mathbf{r} = \left[r_1, r_2, \ldots, r_K \right]$. \\
				\textit{Step 1: Normalization}. For each evaluation $v_{m,k}$, perform the following normalization:
				\begin{equation*}
				u_{m,k} = \frac{v_{m,k}}{\sqrt{\sum_{k=1}^K{v_{m,k}^2}}}, \, \, \, m=1, \ldots, M, \, k=1, \ldots, K.
				\end{equation*}
				\textit{Step 2: Ideal alternatives determination}. Determine both positive and negative ideal alternatives (PIA and NIA, respectively), defined by
				\begin{equation*}
				PIA = \mathbf{u}^+ = \left\{u_{1}^+, u_{2}^+, \ldots, u_{M}^+\right\} \text{and } NIA = \mathbf{u}^- = \left\{u_{1}^-, u_{2}^-, \ldots, u_{M}^-\right\},
				\end{equation*}
				where $u_{m}^+=\max \{u_{m,k} | 1 \leq k \leq K \}$ and $u_{m}^-= \min \{u_{m,k} | 1 \leq k \leq K \}$, $m=1, \ldots, M$. \\
				\textit{Step 3: Covariance matrix}. Determine the covariance matrix $\mathbf{\Sigma}_{\mathbf{U}}=E\left\{\left(\mathbf{U} - E\left\{\mathbf{U}\right\}\right)\left(\mathbf{U} - E\left\{\mathbf{U}\right\}\right)^T\right\}$ of $\mathbf{U}=\left(u_{m,k}\right)$ ($\mathbf{\Sigma}_\mathbf{U} \in \mathbb{R}^{M \times M}$). \\
				\textit{Step 4: Mahalanobis distances}. Calculate the Mahalanobis distances from each alternative and both PIA and NIA:
				\begin{equation*}
				\label{eq:dist_mah_pos}
				DM_k^+ = \sqrt{\left(\mathbf{u}_k-\mathbf{u}^+\right)^T \Delta^T\mathbf{\Sigma}_{\mathbf{U}}^{-1} \Delta \left(\mathbf{u}_k-\mathbf{u}^+\right)}, \, \, \, k=1, \ldots, K
				\end{equation*}
				and
				\begin{equation*}
				\label{eq:dist_mah_neg}
				DM_k^- = \sqrt{\left(\mathbf{u}_k-\mathbf{u}^-\right)^T \Delta^T \mathbf{\Sigma}_{\mathbf{U}}^{-1} \Delta \left(\mathbf{u}_k-\mathbf{u}^-\right)}, \, \, \, k=1, \ldots, K,
				\end{equation*}
				where $\mathbf{u}_k=[u_{1,k}, u_{2,k}, \ldots, u_{M,k}]$ and $\Delta = diag\left(w_1, w_2, \ldots, w_M\right)$ is the diagonal matrix whose elements are the weights $\mathbf{w}$. \\
				\textit{Step 5: Closeness measure}. For each alternative, calculate the closeness measure:
				\begin{equation*}
				r_k = \frac{DM_k^-}{DM_k^+ + DM_k^-}, \, \, \, k=1, \ldots, K,
				\end{equation*}
\end{algorithm}

\begin{theor}{}
Assuming the same importance for all the considered observed criteria (i.e. $w_{m}=w_{m'}$ for all $m$ and $m'$) and that $\mathbf{\Sigma}_{\mathbf{U}}$ is a positive definite matrix, the Mahalanobis distance calculated in Algorithm~\ref{alg:topsis_m} is equivalent to the Euclidean one in a transformed space in which the data are uncorrelated.
\end{theor}

\begin{proof}
Consider the Cholesky factorization~\cite{Golub2013} of $\mathbf{\Sigma}_{\mathbf{U}}$ given by
\begin{equation}
\label{eq:chol_fact}
\mathbf{\Sigma}_{\mathbf{U}} = \mathbf{F}\mathbf{F}^T,
\end{equation}
where $\mathbf{F}$ is an unique lower triangular matrix. In that respect, $DM_k^+$ may be written as
$$DM_k^+ = \sqrt{\left(\mathbf{u}_k-\mathbf{u}^+\right)^T \Delta^T \left(\mathbf{F}\mathbf{F}^T\right)^{-1} \Delta \left(\mathbf{u}_k-\mathbf{u}^+\right)}$$
$$= \sqrt{\left(\Delta \mathbf{u}_k-\Delta \mathbf{u}^+\right)^T \left(\mathbf{F}^T\right)^{-1}\mathbf{F}^{-1} \left(\Delta \mathbf{u}_k-\Delta \mathbf{u}^+\right)}$$
$$= \sqrt{\left(\mathbf{F}^{-1} \Delta \mathbf{u}_k-\mathbf{F}^{-1} \Delta \mathbf{u}^+\right)^T \left(\mathbf{F}^{-1} \Delta \mathbf{u}_k-\mathbf{F}^{-1} \Delta \mathbf{u}^+\right)}$$
$$= \sqrt{\left(\widetilde{\mathbf{u}}_k-\widetilde{\mathbf{u}}^+\right)^T\left(\widetilde{\mathbf{u}}_k-\widetilde{\mathbf{u}}^+\right)}, \, \, \, k=1, \ldots, K,$$
where $\widetilde{\mathbf{u}}_k=\mathbf{F}^{-1} \Delta \mathbf{u}_k$ and $\widetilde{\mathbf{u}}^+=\mathbf{F}^{-1} \Delta \mathbf{u}^+$. Therefore, $DM_k^+$ represents the Euclidean distance between the transformed data $\widetilde{\mathbf{U}}=\left(\widetilde{u}_{m,k}\right)_{M \times K}$ and the transformed positive ideal alternative $\widetilde{\mathbf{u}}^+$. The same conclusion may be achieved considering $DM_k^-$.

In order to show that the transformed data $\widetilde{\mathbf{U}}$ are uncorrelated, consider the covariance matrix
$$\mathbf{\Sigma}_{\widetilde{\mathbf{U}}}=E\left\{\left(\widetilde{\mathbf{U}} - E\left\{\widetilde{\mathbf{U}}\right\}\right)\left(\widetilde{\mathbf{U}} - E\left\{\widetilde{\mathbf{U}}\right\}\right)^T\right\}$$
$$=E\left\{\left(\mathbf{F}^{-1} \Delta \mathbf{U} - \mathbf{F}^{-1} \Delta E\left\{\mathbf{U}\right\}\right)\left(\mathbf{F}^{-1} \Delta \mathbf{U} - \mathbf{F}^{-1} \Delta E\left\{\mathbf{U}\right\}\right)^T\right\}$$
$$=E\left\{\mathbf{F}^{-1} \Delta \left(\mathbf{U} - E\left\{\mathbf{U}\right\}\right)\left(\mathbf{U} - E\left\{\mathbf{U}\right\}\right)^T \Delta^T (\mathbf{F}^{-1})^T \right\}$$
$$=\mathbf{F}^{-1} \Delta E\left\{\left(\mathbf{U} - E\left\{\mathbf{U}\right\}\right)\left(\mathbf{U} - E\left\{\mathbf{U}\right\}\right)^T\right\} \Delta^T (\mathbf{F}^{-1})^T$$
$$=\mathbf{F}^{-1} \Delta {\Sigma}_{\mathbf{U}} \Delta^T (\mathbf{F}^{-1})^T$$
$$=\mathbf{F}^{-1} w_m \mathbf{I}_M {\Sigma}_{\mathbf{U}} w_m^T \mathbf{I}_M^T (\mathbf{F}^{-1})^T$$
$$=w_m^2 \mathbf{F}^{-1} {\Sigma}_{\mathbf{U}} (\mathbf{F}^{-1})^T,$$
where $\mathbf{I}_M$ is the $M \times M$ identity matrix. Since $\mathbf{\Sigma}_{\mathbf{U}} = \mathbf{F}\mathbf{F}^T$, then $\mathbf{F}^{-1}{\Sigma}_{\mathbf{U}}(\mathbf{F}^{-1})^T = \mathbf{I}_M$. Therefore, $\mathbf{\Sigma}_{\widetilde{\mathbf{U}}} = w_m^2\mathbf{I}_M$, i.e. the transformed data $\mathbf{\Sigma}_{\widetilde{\mathbf{U}}}$ is uncorrelated.

\end{proof}

\subsubsection{Graphical interpretation of TOPSIS approaches}
\label{subsub:graphs}

In order to illustrate the problems that arise in TOPSIS due to the correlation between criteria, we provide in this section a graphical interpretation of both TOPSIS and TOPSIS-M approaches. For instance, let us consider a MCDM problem with 200 alternatives and 2 observed dependent criteria, whose evaluations are generated under the light of Equation~\eqref{eq:mix_prop}. We randomly generate the latent criteria according to a uniform distribution in the range $[0,1]$ and apply a mixing matrix given by
$$\mathbf{A}=\left[\begin{array}{cc}
1.00 & 0.70 \\
-0.25 & 1.00 \\
\end{array} \right].$$
Moreover, we consider a weight vector given by $\mathbf{w}=[0.5, 0.5]$.
 
If one applies the TOPSIS in the latent criteria, then one finds the target\footnote{In this paper, we refer to ``target'' PIA and NIA as the ideal solutions that are obtained considering the independent latent criteria. Conversely, we refer to ``target'' ranking as the one provided by TOPSIS applied on the independent latent criteria.} PIA and NIA illustrated in Figure~\ref{fig:gte}. However, if we apply this approach in the observed decision data (mixed data according to~\eqref{eq:mix_prop}), one finds both ideal alternatives illustrated in Figure~\ref{fig:gtem}. It is easy to note that TOPSIS applied on the mixed data does not lead to the desirable PIA and NIA. As a consequence, one does not calculate the correct distance measures, which leads to a ranking of alternatives different from the target one. In other words, the mixing process introduces some bias in the calculation of the PIA and NIA.

\begin{figure}[ht]
\centering
\subfloat[TOPSIS in latent criteria.]{\includegraphics[width=2.6in]{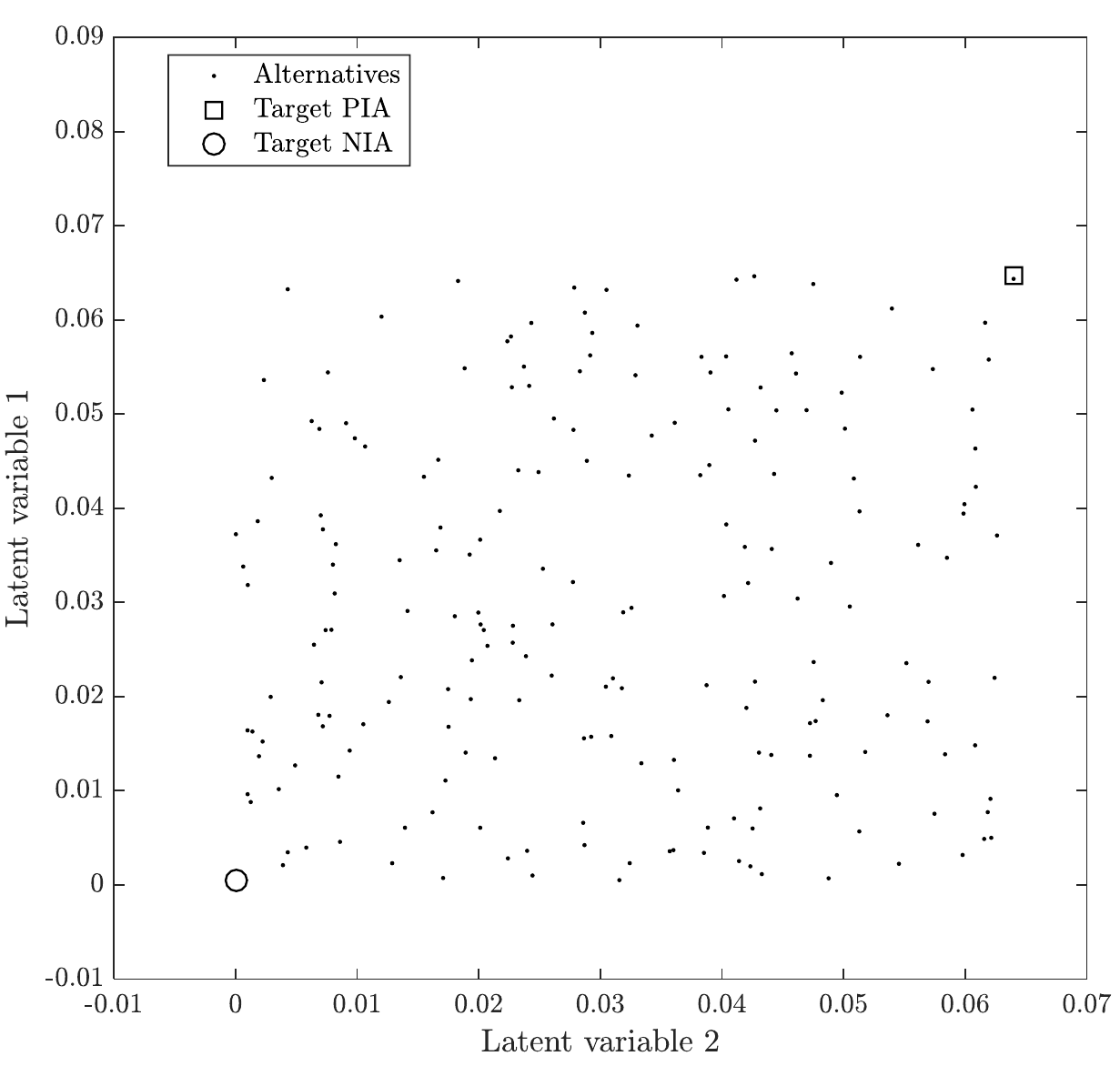}
\label{fig:gte}}
\hfil
\subfloat[TOPSIS in mixed data.]{\includegraphics[width=2.6in]{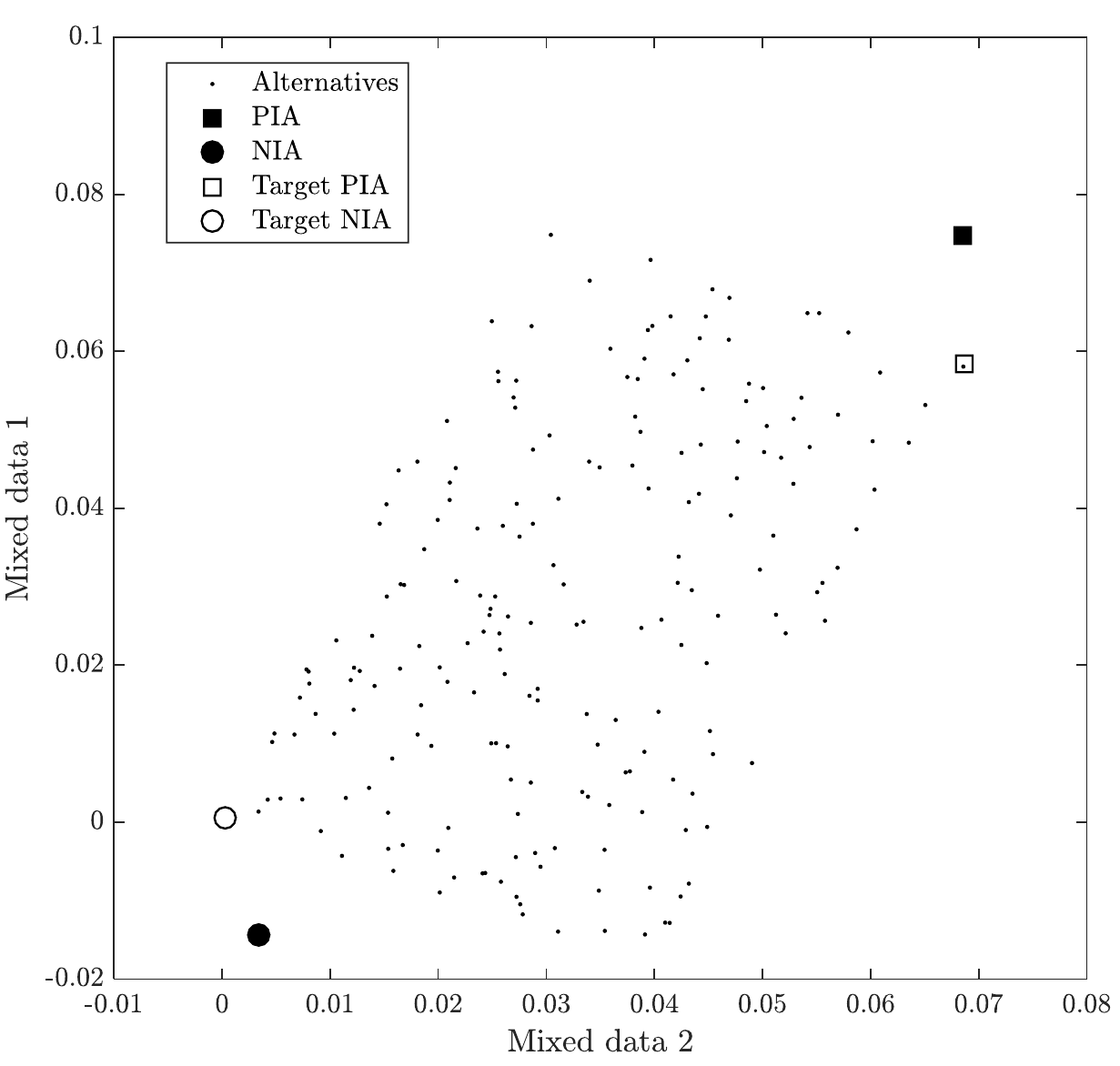}
\label{fig:gtem}}
\hfil
\subfloat[TOPSIS-M in uncorrelated data.]{\includegraphics[width=2.6in]{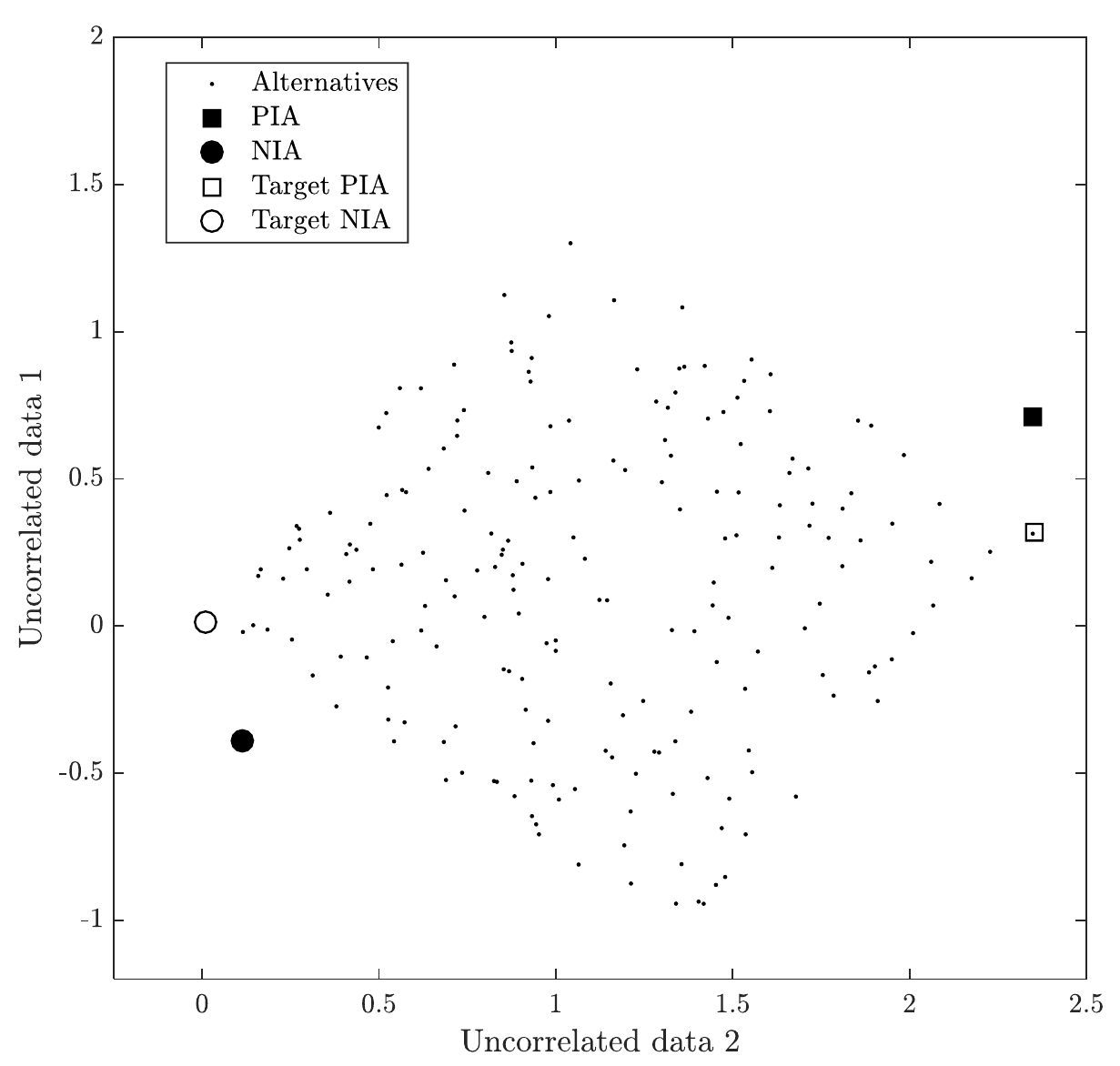}
\label{fig:gtm}}
\caption{Graphical interpretation of TOPSIS and TOPSIS-M approaches.}
\label{fig:gita}
\end{figure}

As discussed in Section~\ref{sec:topsis}, the TOPSIS-M was developed to deal with dependence among criteria in MCDM problems. The application of this approach, which decorrelates the decision data, provides both PIA and NIA illustrated in Figure~\ref{fig:gtm}. Although the data are not correlated anymore, both ideal alternatives are also far from the target ones. This is due to fact that in TOPSIS-M both PIA and NIA ($\widetilde{\mathbf{u}}^+$ and $\widetilde{\mathbf{u}}^-$, respectively) are derived from the transformation of the ideal solutions ($\mathbf{u}^+$ and $\mathbf{u}^-$, respectively) obtained in the mixed data. Therefore, if PIA and NIA in mixed data are not the target ones, the application of TOPSIS-M will not lead to the target ideal solutions.

The aforementioned results point out that the covariance information among criteria and the transformation procedure conducted by TOPSIS-M may not be enough to mitigate the biased effect provided by dependent criteria. Since decorrelation does not imply in independence~\cite{Papoulis2002}, a method that aims at retrieving a set of independent data may lead to better results. We discuss a possible approach in the next section.

\subsection{Blind source separation}
\label{sec:bss}

In summary, blind source separation problems~\cite{Comon2010} consist in retrieving a set of signal sources $\mathbf{s}(k)=[s_1(k), s_2(k), \ldots, s_N(k)]$ given a set of mixtures $\mathbf{x}(k)=[x_1(k), x_2(k), \ldots, x_M(k)]$ of these sources, which is obtained (in the instantaneous linear case) by
\begin{equation}
\label{eq:mix}
\mathbf{x}(k) = \mathbf{A}\mathbf{s}(k) + \mathbf{g}(k),
\end{equation}
where $\mathbf{A} = \left( a_{m,n} \right) \in \mathbb{R}^{M \times N}$ is the mixing matrix and $\mathbf{g}(k)=[g_1(k), g_2(k), \ldots, g_M(k)]$ models an additive white Gaussian noise. The problem is called ``blind'' since both $\mathbf{A}$ and $\mathbf{s}(k)$ are unknown. Therefore, based only in the set of mixtures $\mathbf{x}(k)$ and in a prior information about the signal sources (statistical independence, for example), the aim is to adjust a separating matrix $\mathbf{B} = \left( b_{n,m} \right) \in \mathbb{R}^{N \times M}$ that provides a set of estimates $\mathbf{y}(k)=[y_1(k), y_2(k), \ldots, y_N(k)]$, given by
\begin{equation}
\label{eq:demix}
\mathbf{y}(k) = \mathbf{B}\mathbf{x}(k),
\end{equation}
which should be as close as possible from $\mathbf{s}(k)$. Figure~\ref{fig:problem_bss} illustrates this procedure. Ideally, we expect $\mathbf{B}$ to be equal to the inverse of $\mathbf{A}$, i.e. $\mathbf{B}=\mathbf{A}^{-1}$. However, one may achieve estimates that are different from the signal sources by permutation ($y_n(k) \approx s_{n'}(k)$, for $n \neq n'$) and/or scaling ambiguities ($y_n(k) \approx \lambda s_{n'}(k)$, for a scalar $\lambda \in \mathbb{R}$)~\cite{Comon2010}. The procedure used to deal with these ambiguities, which is typical when applying BSS methods, will be addressed in Section~\ref{subsec:icatopsis}.

\begin{figure}[ht]
\centering
\includegraphics[height=3.8cm]{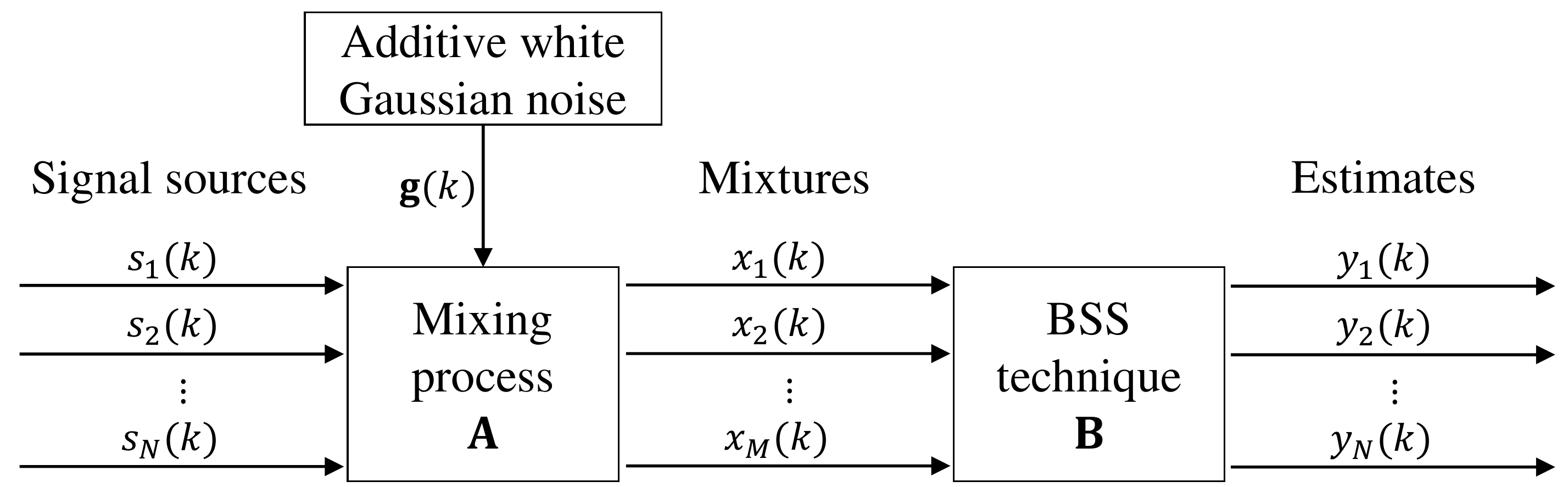}
\caption{Blind source separation problem.}
\label{fig:problem_bss}
\end{figure}

In the literature, one may find several approaches to deal with BSS problems~\cite{Comon2010}. Among them, a common one is independent component analysis (ICA), which main idea is to assume that the sources can be modeled as samples of independent and identically distributed (i.i.d.) and non-Gaussian random variables~\cite{Hyvarinen2001}. ICA is based on the observation that the mixtures provided by~\eqref{eq:mix} are not independent anymore, since it represents a linear combination of the sources.

In this paper, we consider two ICA algorithms to deal with the MCDM problem, called FastICA and JADE. FastICA performs an optimization procedure which separation criterion can be seen as a simplified route to achieve independence~\cite{Hyvarinen2001}. An example is the kurtosis, defined by
\begin{equation}
kurt(y) = E\left\{y^4 \right\} - 3\left(E\left\{y^2 \right\}\right),
\end{equation}
which represents the fourth-order moment of a random variable $y$. For further details about FastICA algorithm, see~\cite{Hyvarinen2001}.

The second ICA technique considered in this paper, called JADE (Joint Approximate Diagonalization of Eigenmatrices)~\cite{Cardoso1993}, considers the diagonalization of a set of forth order cumulant matrices associated with the retrieved sources to achieve independence among the retrieved signals. It is worth mentioning that, in order to apply JADE algorithm (as well as FastICA), a previous step is necessary, which comprises a decorrelation of the mixed signals~\cite{Hyvarinen2001}. Further details about JADE algorithm can be found in~\cite{Cardoso1993}. 

\section{Proposed approaches}
\label{sec:proposal}

As already mentioned in this paper, we propose two different approaches to deal with dependent criteria in MCDM problems, both based on BSS techniques. They are presented in the sequel.

\subsection{ICA-TOPSIS}
\label{subsec:icatopsis}

The first approach considered in this paper, called ICA-TOPSIS~\cite{Pelegrina2018}, comprises three steps, as described in Algorithm~\ref{alg:ica_topsis}. In the first one, we assume that the decision data $\mathbf{V}$ is composed by a mixture of latent criteria $\mathbf{L}$ according to the mixing process described in~\eqref{eq:mix_prop}. Therefore, we formulate a BSS problem and apply an ICA technique to adjust the separating matrix $\mathbf{B}$, which gives us the estimates $\mathbf{\hat{L}}$ (alternative representation of data $\mathbf{V}$) according to Equation~\eqref{eq:demix_mcdm}\footnote{If we consider the separating process described in~\eqref{eq:demix}, $\mathbf{y}(k)$ and $\mathbf{x}(k)$ represent, respectively, $\mathbf{\hat{L}}$ and $\mathbf{V}$. Moreover, we consider the determined case in BSS, in which the number of observed criteria is equal to the number of independent latent criteria ($M=N$).}.
 
\begin{algorithm}
%\setstretch{1.00}
%\SetAlgoLined
    \caption{ICA-TOPSIS}
    \label{alg:ica_topsis}
				\textbf{Input:} Decision data $\mathbf{V}$ and set of weights $\mathbf{w}$. \\
				\textbf{Output:} Closeness measure $\mathbf{r} = \left[r_1, r_2, \ldots, r_K \right]$. \\
				\textit{Step 1: Latent criteria estimation}. Based on $\mathbf{V}$, apply an ICA technique aiming at retrieving the latent criteria $\mathbf{L}$. \\
				\textit{Step 2: Ambiguities mitigation}. Perform adjustments in both estimated mixing matrix $\mathbf{\hat{A}}$ and estimated latent criteria $\mathbf{\hat{L}}$ derived from the ICA technique in order to avoid permutation and/or scaling ambiguities. \\
				\textit{Step 3: TOPSIS and ranking determination}. Apply theTOPSIS approach (Algorithm~\ref{alg:topsis_e}), using the adjusted estimated latent criteria $\mathbf{\hat{L}}^{Adj_{c}}$ and the set of weights $\mathbf{w}$ as inputs, and determine the ranking of alternatives.
\end{algorithm}

As mentioned in Section~\ref{sec:bss}, BSS methods may suffer from permutation and/or scaling ambiguities. In this context, in the second step, we perform adjustments in the estimated latent criteria $\mathbf{\hat{L}}$ in order to avoid these ambiguities. In that respect, we assume that (i) the latent criteria $\mathbf{L}$ are in similar range and (ii) the diagonal elements in the mixing matrix $\mathbf{A}$ are positive and greater (in absolute value) than the off-diagonal elements in the same row, which is equivalent to consider that each latent data has a positive majority influence in each observed criterion. If we consider the example illustrated in Figure~\ref{fig:problem_example}, competences in social sciences, natural sciences and mathematics (independent latent criteria) should have a positive majority influence in subjects such as human resources, physics and linear algebra, respectively. Therefore, based on the mixing matrix $\mathbf{\hat{A}}=\mathbf{B}^{-1}$ derived by an ICA technique, the adjustments performed in $\mathbf{\hat{L}}$ are the following ones:
\begin{enumerate}
\item For the first row in $\mathbf{\hat{A}}$, we find the column $q$ that contains the greater absolute value. This value represents the influence of the estimated latent criterion $\hat{l}_q$ in the mixed process that results in the set of evaluations $v_{1,k}$ ($k=1, 2, \ldots, K$) with respect to the first criterion (first column of $\mathbf{V}$). Therefore, in order to place $\hat{l}_q$ as the first estimated latent criterion (without invaliding the relation $\mathbf{V}=\mathbf{\hat{A}}\mathbf{\hat{L}}$), we permute the first and the $q$ columns of $\mathbf{\hat{A}}$ and also the first and the $q$ estimates. This procedure is repeated for all rows $\mathbf{\hat{A}}$, which leads to both estimated mixing matrix ($\mathbf{\hat{A}}^{Adj_{p}}$) and estimated latent criteria ($\mathbf{\hat{L}}^{Adj_{p}}$) partially adjusted, and mitigates the permutation ambiguity provided by the BSS method.
\item For each column in $\mathbf{\hat{A}}^{Adj_{p}}$, we verify the signal of the diagonal element. If the diagonal element $q'$ is negative (negative contribution of the associated estimated latent criterion $\hat{l}_{q'}$), we multiply all the elements in the same column of $q'$ by $-1$. As a consequence, we also need to invert the signal of $\hat{l}_{q'}$, since it is necessary to ensure that $\mathbf{V}=\mathbf{\hat{A}}\mathbf{\hat{L}}$. After verifying all the diagonal elements of $\mathbf{\hat{A}}^{Adj_{p}}$ and performing the signal changes, we obtain both estimated mixing matrix ($\mathbf{\hat{A}}^{Adj_{c}}$) and estimated latent criteria ($\mathbf{\hat{L}}^{Adj_{c}}$) completely adjusted, which mitigates the signal inversion ambiguity provided by the BSS method. With respect to the scaling ambiguity provided by a positive factor or a negative factor different from $-1$, this is automatically mitigated in the normalization step of TOPSIS.
\end{enumerate}

Finally, after applying ICA and eliminating permutation and/or scaling ambiguities, the third step of the proposed approach comprises the application of TOPSIS in $\mathbf{\hat{L}}^{Adj_{c}}$ in order to derive the ranking of alternatives. Therefore, Steps 1 and 2 can be seen as a preprocessing step.

In order to illustrate the application of the ICA-TOPSIS approach, consider the example described in Section~\ref{subsub:graphs}. After applying the ICA algorithm, we obtain the estimated mixing matrix
$$\mathbf{\hat{A}}=\left[\begin{array}{cc}
-0.2158 & 0.2864 \\
-0.2888 & -0.0651 \\
\end{array} \right],$$
which is associated with the estimated latent criteria $\mathbf{\hat{L}}=[\hat{l}_1, \hat{l}_2]^T$. In order to mitigate the permutation ambiguity, the first adjustment leads to both estimated mixing matrix
$$\mathbf{\hat{A}}^{Adj_{p}}=\left[\begin{array}{cc}
0.2864 & -0.2158 \\
-0.0651 & -0.2888 \\
\end{array} \right]$$
and estimated latent criteria $\mathbf{\hat{L}}^{Adj_{p}}=[\hat{l}_2, \hat{l}_1]$ partially adjusted, which comprises the permutation of both columns of $\mathbf{\hat{A}}^{Adj_{p}}$ and estimates $\hat{l}_1$ and $\hat{l}_2$. With respect to the scaling ambiguity, the second adjustment leads to both estimated mixing matrix\footnote{One may note that this estimated mixing matrix differs from the correct one (defined on Section~\ref{subsub:graphs}) by a factor of 3.5, approximately, which introduces a scaling ambiguity provided by a positive factor.}
$$\mathbf{\hat{A}}^{Adj_{c}}=\left[\begin{array}{cc}
0.2864 & 0.2158 \\
-0.0651 & 0.2888 \\
\end{array} \right]$$
and estimated latent criteria $\mathbf{\hat{L}}^{Adj_{c}}=[\hat{l}_2, -\hat{l}_1]$ completely adjusted, which comprises the signal inversion of both second column of $\mathbf{\hat{A}}^{Adj_{c}}$ and second estimates of $\mathbf{\hat{L}}^{Adj_{c}}$. Therefore, based on $\mathbf{\hat{L}}^{Adj_{c}}$, we apply the TOPSIS approach, which leads to the PIA and NIA illustrated in Figure~\ref{fig:git}. One may note that the ICA algorithm performs a good estimation of the latent criteria, which also provides both PIA and NIA close to the target ones.

\begin{figure}[ht]
\centering
\subfloat[TOPSIS in latent criteria.]{\includegraphics[width=2.6in]{gra_topsis_e.pdf}
\label{fig:gte2}}
\hfil
\subfloat[ICA-TOPSIS in estimated latent criteria.]{\includegraphics[width=2.6in]{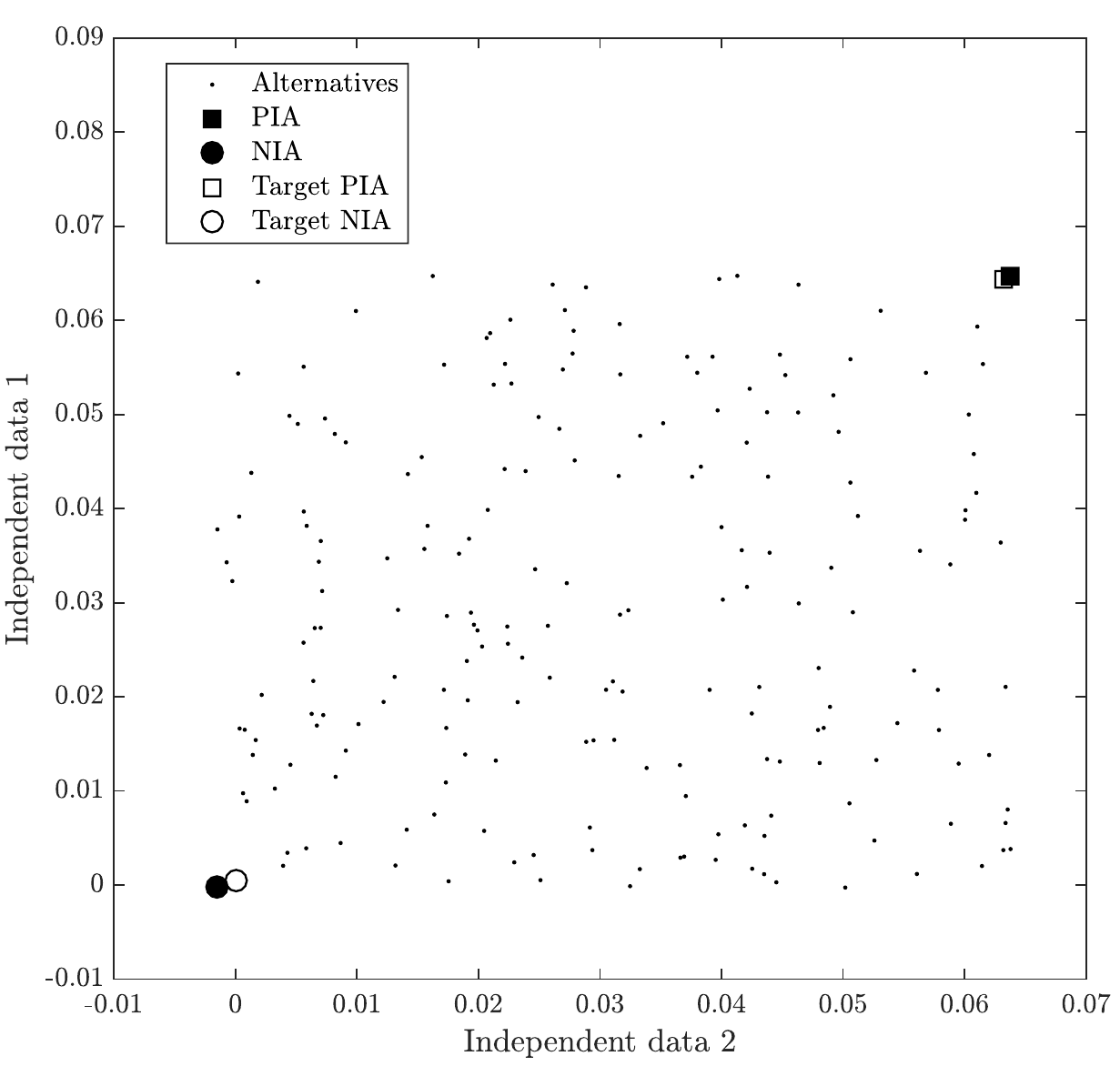}
\label{fig:git}}
\hfil
\subfloat[ICA-TOPSIS-M in uncorrelated data.]{\includegraphics[width=2.6in]{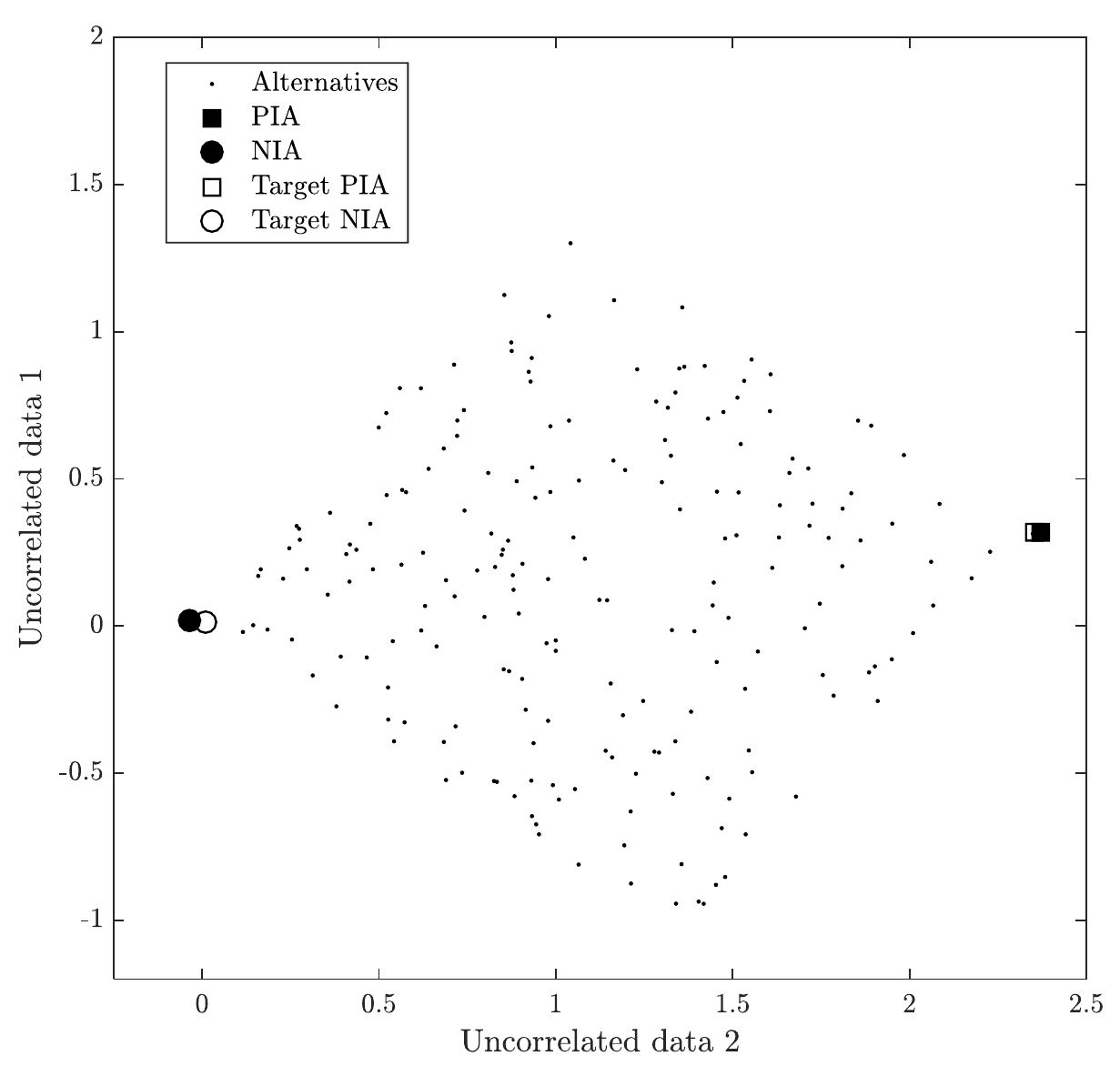}
\label{fig:gitm}}
\caption{Graphical interpretation of TOPSIS, ICA-TOPSIS and ICA-TOPSIS-M approaches.}
\label{fig:gita2}
\end{figure}

\subsection{ICA-TOPSIS-M}
\label{subsec:icatopsism}

The second approach considered in this paper, called ICA-TOPSIS-M, combines an ICA technique and a modified version of the TOPSIS-M. It comprises five steps, as described in Algorithm~\ref{alg:ica_topsis_m}. The first and the second steps are identical to the ICA-TOPSIS approach. However, the third step comprises the determination of both PIA and NIA directly from the estimated latent criteria completely adjusted $\mathbf{\hat{L}}^{Adj_{c}}$, without any normalization. In the fourth step, we apply the estimated mixing matrix completely adjusted $\mathbf{\hat{A}}^{Adj_{c}}$ on both PIA and NIA, which leads to the transformation of these ideal alternatives into the mixed space (the same as the mixed decision data $\mathbf{V}$). Since the transformed PIA and NIA were obtained based on the estimated latent criteria, we expect that both ideal alternatives are placed close to the target ones in the mixed space.

\begin{algorithm}
%\setstretch{1.00}
%\SetAlgoLined
    \caption{ICA-TOPSIS-M}
    \label{alg:ica_topsis_m}
				\textbf{Input:} Decision data $\mathbf{V}$ and set of weights $\mathbf{w}$. \\
				\textbf{Output:} Closeness measure $\mathbf{r} = \left[r_1, r_2, \ldots, r_K \right]$. \\
				\textit{Step 1: Latent criteria estimation}. Based on $\mathbf{V}$, apply an ICA technique aiming at retrieving the latent criteria $\mathbf{L}$. \\
				\textit{Step 2: Ambiguities mitigation}. Perform adjustments in both estimated mixing matrix $\mathbf{\hat{A}}$ and estimated latent criteria $\mathbf{\hat{L}}$ derived from the ICA technique in order to avoid permutation and/or scaling ambiguities. \\
				\textit{Step 3: Ideal alternatives determination}. Determine both positive and negative ideal alternatives (PIA and NIA, respectively) directly from $\mathbf{\hat{L}}^{Adj_{c}}$:
				\begin{equation*}
				PIA = \mathbf{\hat{l}}^+ = \left\{\hat{l}_{1}^+, \hat{l}_{2}^+, \ldots, \hat{l}_{N}^+\right\} \text{and } NIA = \mathbf{\hat{l}}^- = \left\{\hat{l}_{1}^-, \hat{l}_{2}^-, \ldots, \hat{l}_{N}^-\right\},
				\end{equation*}
				where $\hat{l}_{n}^+=\max \{\hat{l}_{n,k} | 1 \leq k \leq K \}$ and $\hat{l}_{n}^-= \min \{\hat{l}_{n,k} | 1 \leq k \leq K \}$, $n=1, \ldots, N$. \\
				\textit{Step 4: Ideal alternatives transformation}. Apply $\mathbf{\hat{A}}^{Adj_{c}}$ on both PIA and NIA in order to transform these ideal alternatives into the mixed data space. \\
				\textit{Step 5: Modified TOPSIS-M and ranking determination}. Based on the transformed PIA and NIA calculated in Step 4, apply the modified TOPSIS-M approach (Algorithm~\ref{alg:topsis_m} without the Step 3) and determine the ranking of alternatives.
\end{algorithm}

Finally, the last step comprises the application of a modified version of TOPSIS-M approach to determine the ranking of alternatives. For instance, instead of deriving PIA and NIA in Step 3 of Algorithm~\ref{alg:topsis_m}, we use as input the transformed PIA and NIA calculated in the Step 4 of the ICA-TOPSIS-M approach. Therefore, we expect that these ideal alternatives be placed close to the target ones, which leads to an improvement in the distances calculation of TOPSIS-M approach (Step 3 of Algorithm~\ref{alg:topsis_m}) and, therefore, in the ranking determination.

Considering the example described in Section~\ref{subsub:graphs} and applying ICA-TOPSIS-M approach, with the transformed PIA and NIA derived from $\mathbf{\hat{L}}^{Adj_{c}}$, we obtain the ideal alternatives illustrated in Figure~\ref{fig:gitm}. Since these ideal alternatives are derived directly from the estimated latent criteria instead of the observed mixed data, they are closer to the target PIA and NIA compared with the ones obtained by the TOPSIS-M approach, illustrated in Figure~\ref{fig:gtm}.

\section{Experiments and results}
\label{sec:exper}

In to order to test the proposed approaches in MCDM problems with dependent criteria, in this section, we perform experiments in both synthetic and real data. They are described in the sequel.

\subsection{Experiments on synthetic data}
\label{subsec:exper_synt}

In order to verify the robustness of the proposed approaches compared to the traditional TOPSIS and TOPSIS-M methods, we performed numerical experiments on synthetic data, whose latent criteria were randomly generated according to a uniform distribution in the range $[0,1]$. This scenario may represent the example described in Figure~\ref{fig:problem_example}, in which the latent criteria and mixed decision data correspond, respectively, to the competences of students in different subjects (to be estimated) and their grades. Moreover, we considered that all criteria have the same importance on the decision data, i.e. $w_{m} = w_{m'}$ for all $m$ and $m'$.

\subsubsection{TOPSIS-M performance for different mixing matrices}
\label{subsubsuc:exp1}

The first experiment comprises the analysis of TOPSIS-M approach to deal with the problem addressed in this paper for different values of the off-diagonal elements of the mixing matrix. Therefore, based on two latent criteria randomly generated (with 100 alternatives), we generate the mixed decision data, as described in Equation~\eqref{eq:mix_prop} without additive noise, using the mixing matrix
$$\mathbf{A}=\left[\begin{array}{cc}
1 & \alpha \\
\beta & 1 \\
\end{array} \right],$$
where the off-diagonal elements $\alpha$ and $\beta$ assume values in the range $[-0.75,0.75]$. In order to verify the robustness of TOPSIS-M approach to deal with different mixture intensities, we considered as a performance index the Kendall tau coefficient~\cite{Kendall1938}, defined by
\begin{equation}
\tau=\frac{J_{conc}-J_{disc}}{K(K-1)/2},
\end{equation}
where $J_{conc}$ and $J_{disc}$ are the number of pairwise agreements and pairwise disagreements between two rankings, respectively, $K$ is the number of alternatives and $K(K-1)/2$ is the total number of pairwise combinations. $\tau=1$ and $\tau=-1$ indicate that the two ranking are the same and the reverse of the other, respectively. If $\tau \approx 0$, the two rankings are independent. Therefore, in our experiments, the larger $\tau$ the better, since it indicates that the retrieved ranking is close to the target one provided by the application of TOPSIS on the latent criteria.

After applying TOPSIS-M approach on mixed decision data, the obtained Kendall tau coefficients (averaged over 500 simulations) are shown in Figure~\ref{fig:rob_mahal}. One may note that TOPSIS-M approach achieves better results when both $\alpha$ and $\beta$ are positive or both $\alpha$ and $\beta$ are negative and equal, i.e. where there is no negative contribution of only one latent criterion in the mixing process. However, if $\alpha$ and/or $\beta$ are negative (and different), $\tau$ decreases, indicating that there are a great number of disagreements between the ranking obtained by the TOPSIS-M approach applied on the mixed decision data and the target one provided by TOPSIS applied on the latent criteria. This is a consequence of the fact that both PIA and NIA are far from the target values.

\begin{figure}[ht]
\centering
\includegraphics[height=8.5cm]{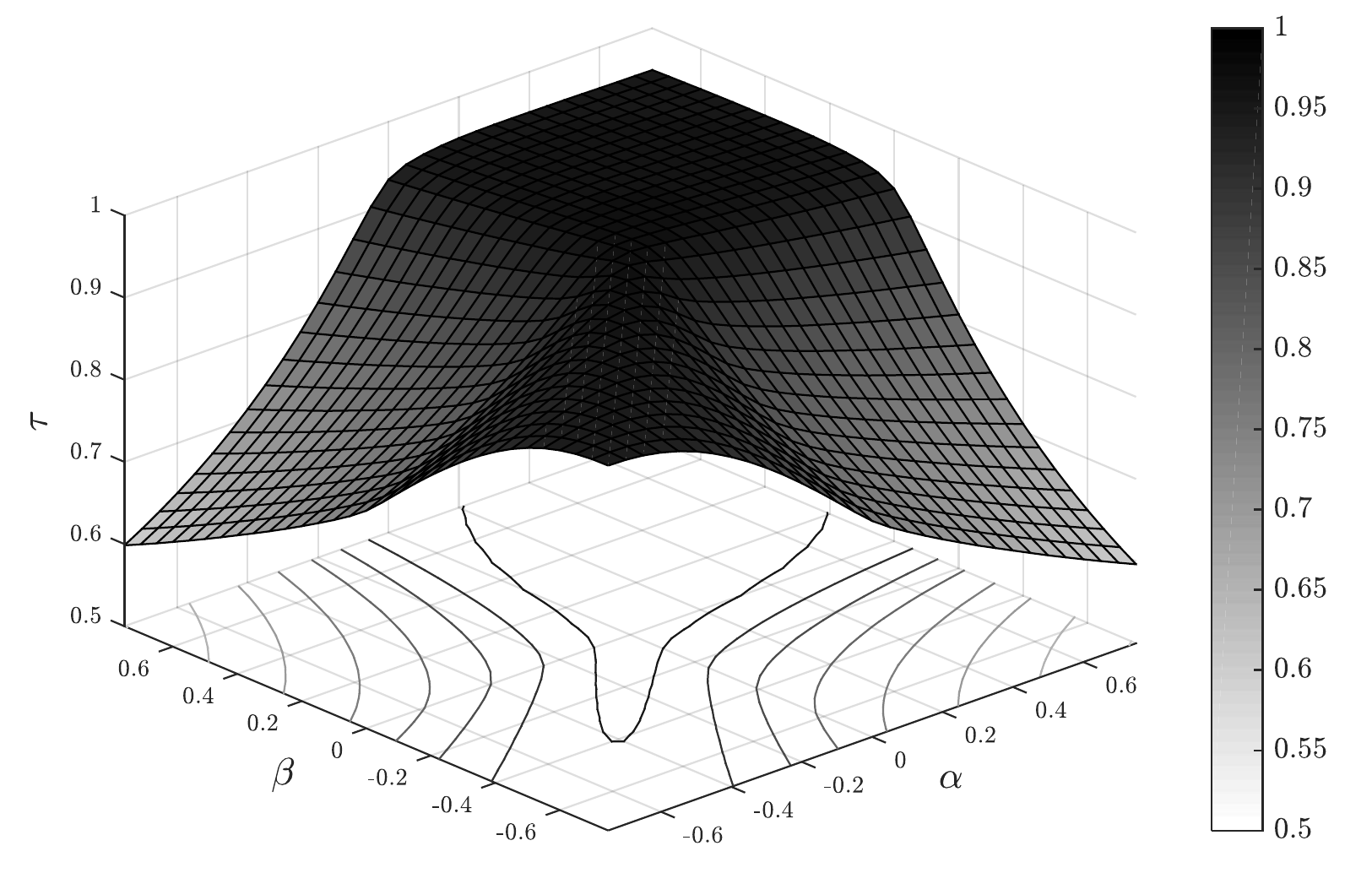}
\caption{Robustness of TOPSIS-M for different degrees of mixture.}
\label{fig:rob_mahal}
\end{figure}

\subsubsection{Comparison for different degrees of noise}
\label{subsubsec:exp2}

In this second experiment, we compare the performance of TOPSIS, TOPSIS-M and the proposed ICA-TOPSIS and ICA-TOPSIS-M approaches. They are applied to deal with a MCDM problem with $K=100$ alternatives, $M=2$ criteria and for different degrees of noise $\mathbf{G}$ in the mixing process~\eqref{eq:mix_prop}. This analysis is important since it points out how robust the approaches are in situations in which the generative model adopted for the observed data is not perfect (which is often the case in real situations). In our experiment, we provide results for different Signal-to-Noise Ratio (SNR, in dB), given by
\begin{equation}
SNR = 10 \log_{10} \frac{P_{signal}^2}{P_{noise}^2},
\end{equation}
where $P_{signal}^2$ and $P_{noise}^2$ are, respectively, the signal ($\mathbf{AL}$ in the mixing process~\eqref{eq:mix_prop}) power and the noise ($\mathbf{G}$) power, in the range $(0,50]$ dB. Therefore, lower values of SNR indicates lower degrees of noise in the mixing process.

Figure~\ref{fig:snr2k} presents a comparison of the obtained Kendall tau coefficients (averaged over 500 simulations for each SNR value) for the considered approaches (with ICA-TOPSIS and ICA-TOPSIS-M using both FastICA and JADE algorithms). Figure~\ref{fig:snr2kit} also illustrates $\tau$ for the ``Utopic ICA-TOPSIS'', which comprises the application of ICA-TOPSIS with the ideal separating matrix $\mathbf{B}=\mathbf{A}^{-1}$ and is used as benchmark for the achieved results. Similarly, in Figure~\ref{fig:snr2kitm}, we use as a benchmark the ``Utopic ICA-TOPSIS-M'', which comprises the application of ICA-TOPSIS-M with the transformation of the target PIA and NIA derived from the latent criteria.

\begin{figure}[ht]
\centering
\subfloat[ICA-TOPSIS approach]{\includegraphics[width=4.5in]{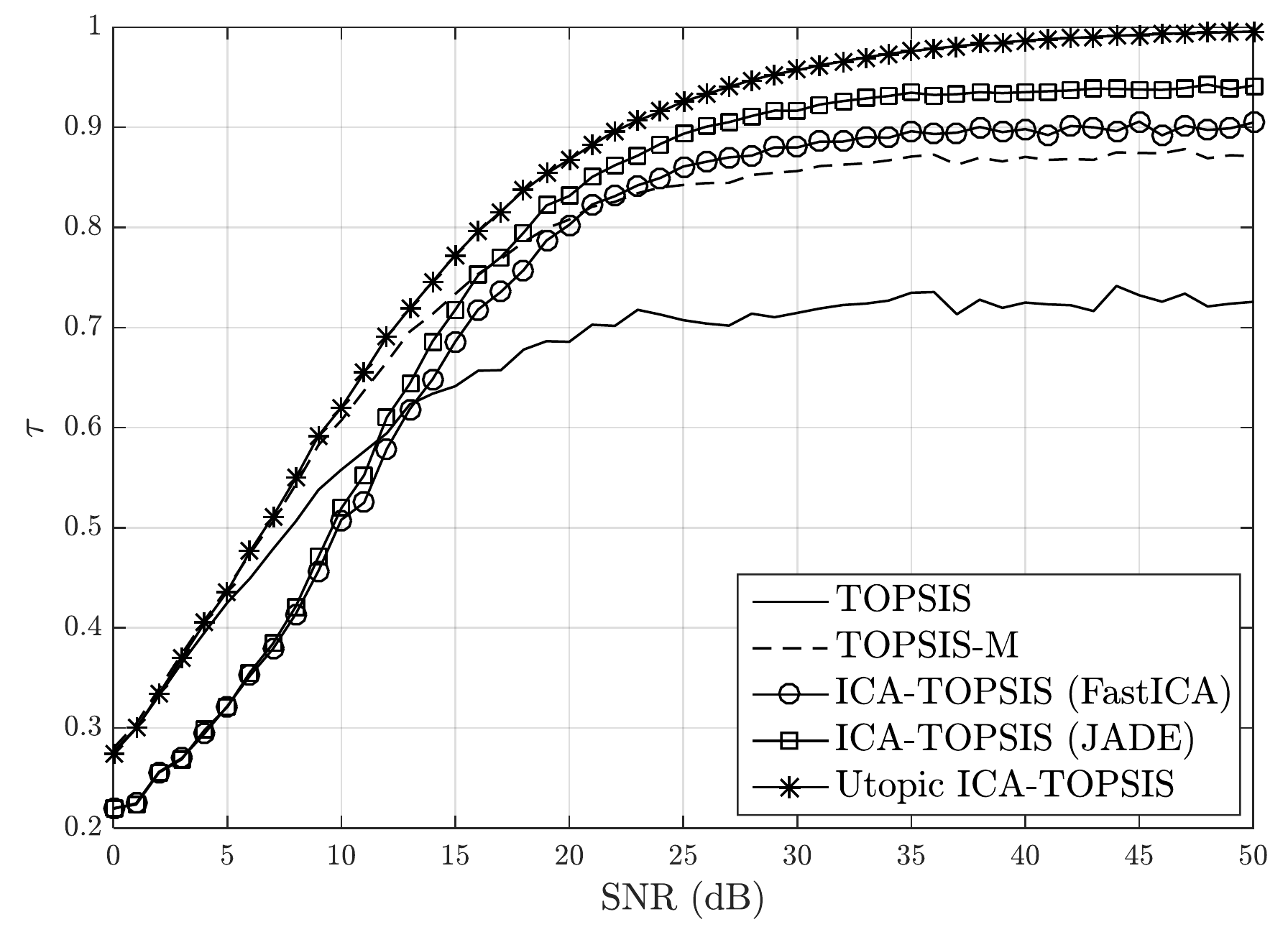}
\label{fig:snr2kit}}
\hfil
\subfloat[ICA-TOPSIS-M approach]{\includegraphics[width=4.5in]{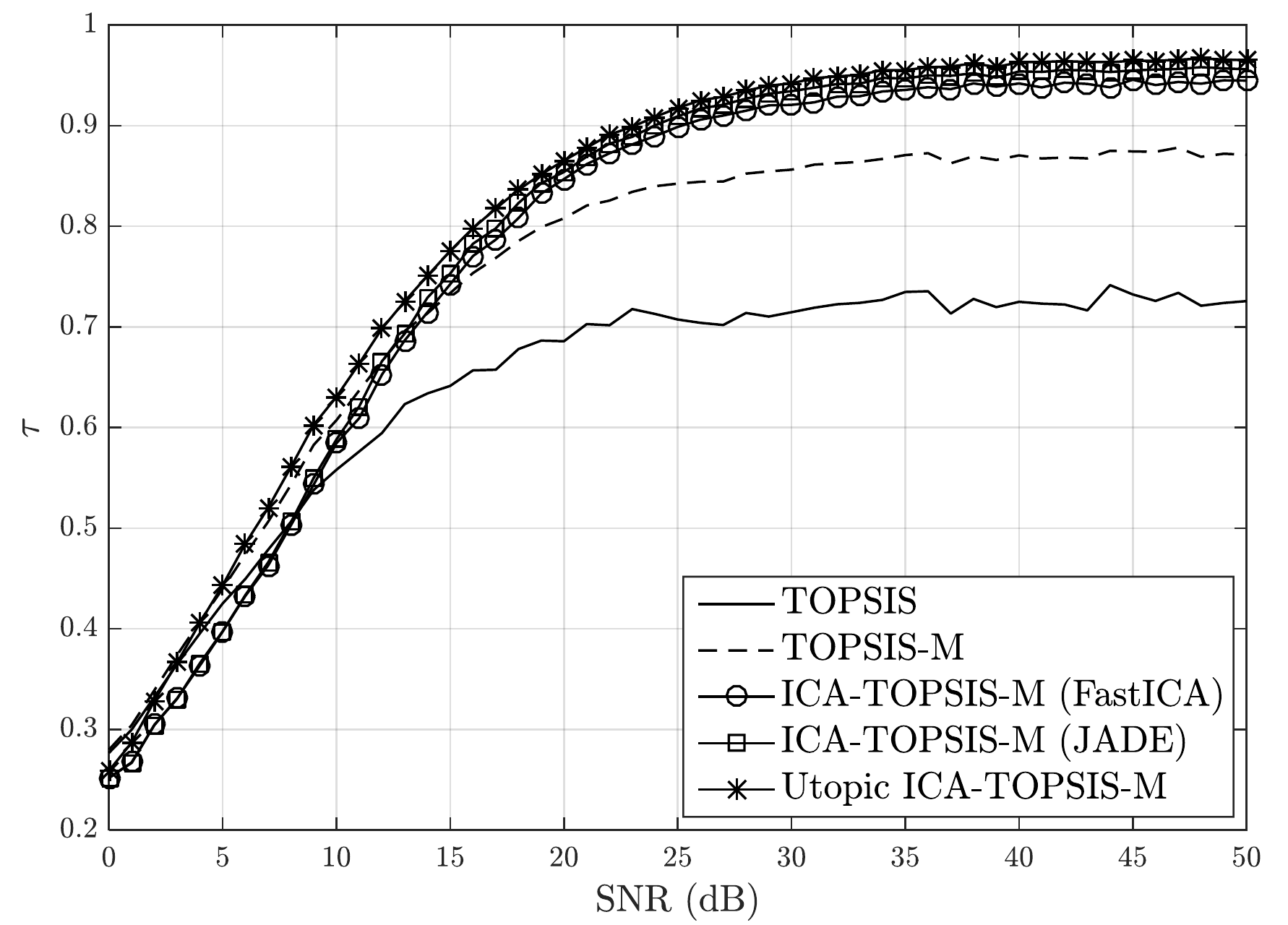}
\label{fig:snr2kitm}}
\caption{Comparison of Kendall tau coefficients for different SNR values.}
\label{fig:snr2k}
\end{figure}

Although the TOPSIS-M approach applied to the mixed decision data performed better compared to the TOPSIS, an improvement may be achieved using either ICA-TOPSIS or ICA-TOPSIS-M, specially for SNR greater than 15 dB (ICA-TOPSIS-M performed slightly better than ICA-TOPSIS). For SNR lower than 15 dB, the interference of noise is strong enough to hinder the considered approaches to retrieve a proper ranking of alternatives.

In order to further exploit the performance of the considered approaches, we calculate the Pearson's correlation coefficient between the closeness measure of the estimates ($r_k^{\mathbf{\hat{L}}^{Adj_{c}}}$) and the independent latent criteria ($r_k^{\mathbf{L}}$), given by
\begin{equation}
\rho \left(r_k^{\mathbf{\hat{L}}^{Adj_{c}}}, r_k^{\mathbf{L}} \right) = \frac{\sum_{k=1}^K ( r_k^{\mathbf{\hat{L}}^{Adj_{c}}} - \bar{r}_k^{\mathbf{\hat{L}}^{Adj_{c}}} ) ( r_k^{\mathbf{L}} - \bar{r}_k^{\mathbf{L}} )}{\sqrt{\sum_{k=1}^K ( r_k^{\mathbf{\hat{L}}^{Adj_{c}}} - \bar{r}_k^{\mathbf{\hat{L}}^{Adj_{c}}} )^2} \sqrt{\sum_{k=1}^K ( r_k^{\mathbf{L}} - \bar{r}_k^{\mathbf{L}} )^2}},
\end{equation}
where $\bar{r}_k^{\cdot} = \left(1/K \right) \sum_{k=1}^K r_k^{\cdot}$. Moreover, we also calculates the mean absolute error of the position of the first 20\% of the total number ($K$) of alternatives (the first 20 alternatives, in this case). This measure is given by:
\begin{equation}
\varepsilon = \frac{1}{0.2K}\sum_{k=1}^{0.2K}\left| o_k - k \right|,
\end{equation}
where $o_k$ is the position, in the ranking derived by the considered approaches, of the $k$-th alternative in the correct ranking. Figures~\ref{fig:snr2c} and~\ref{fig:snr2e} presents the obtained results. Similarly to the Kendall tau coefficient, one may note that both ICA-TOPSIS and ICA-TOPSIS-M approaches lead to better results, specially for SNR greater than 25 dB. One also may note that ICA-TOPSIS-M performed slightly better than ICA-TOPSIS.

\begin{figure}[ht]
\centering
\subfloat[ICA-TOPSIS approach]{\includegraphics[width=4.5in]{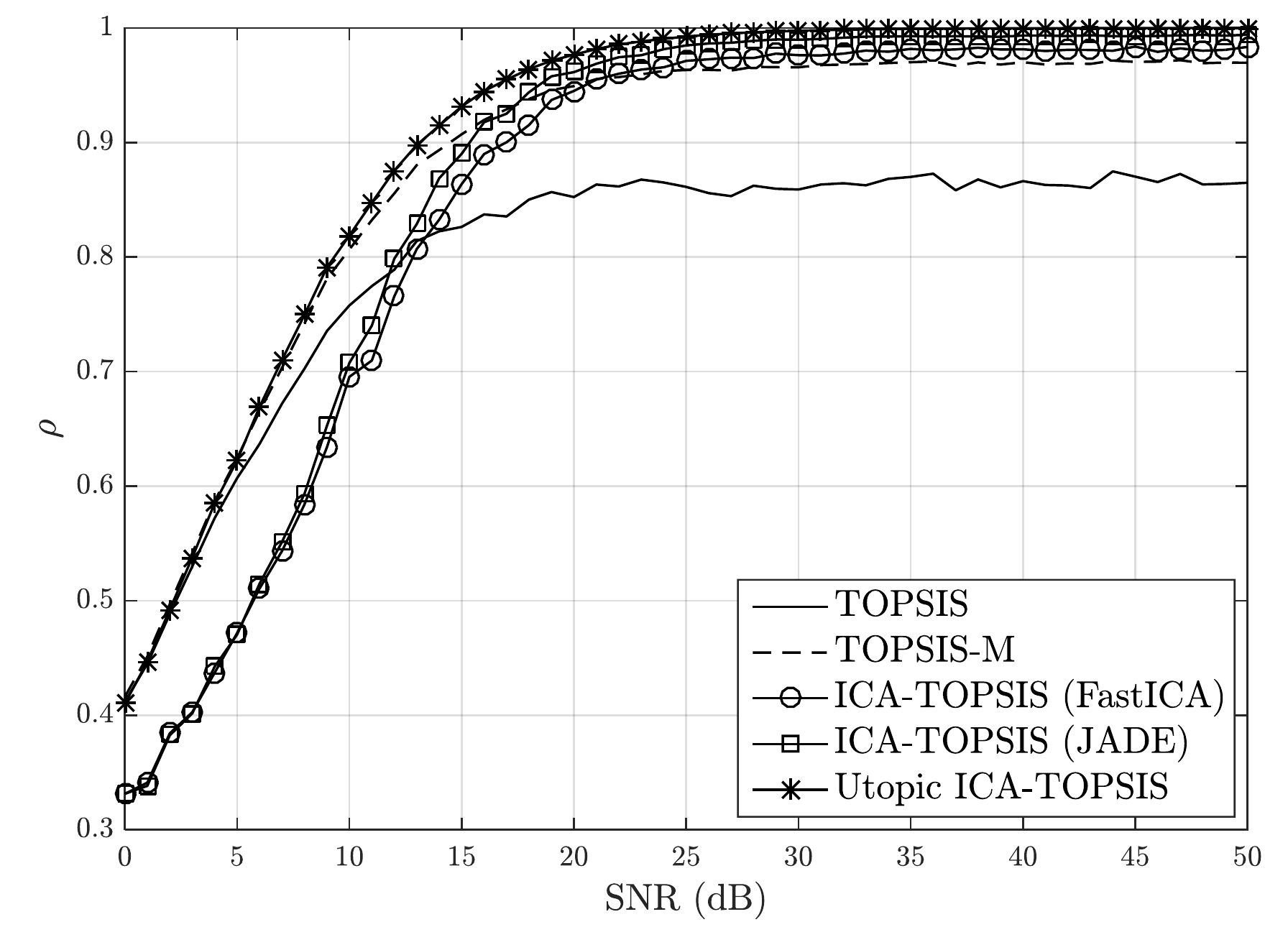}
\label{fig:snr2cit}}
\hfil
\subfloat[ICA-TOPSIS-M approach]{\includegraphics[width=4.5in]{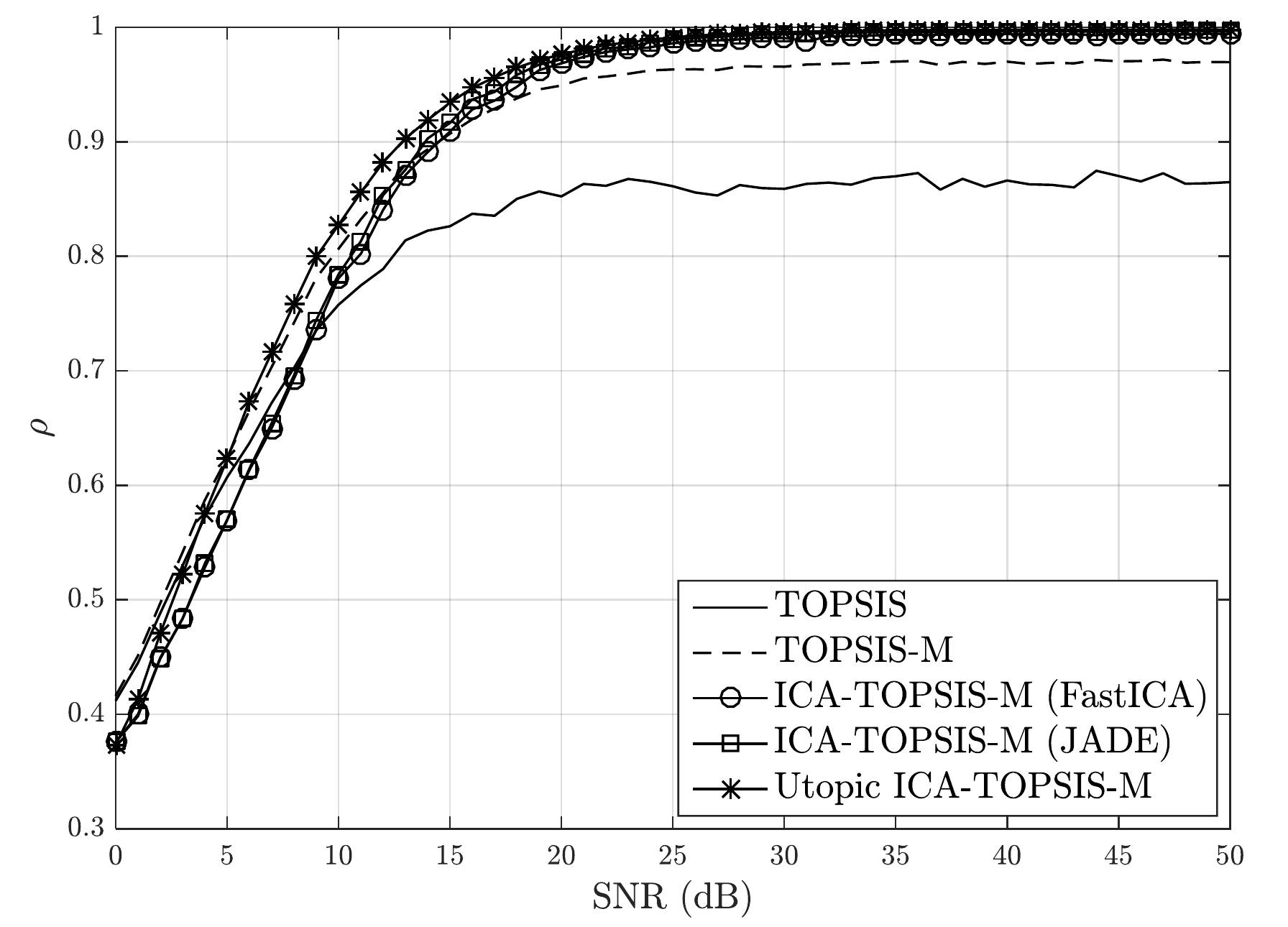}
\label{fig:snr2citm}}
\caption{Comparison of Pearson's correlation coefficients for different SNR values.}
\label{fig:snr2c}
\end{figure}

\begin{figure}[ht]
\centering
\subfloat[ICA-TOPSIS approach]{\includegraphics[width=4.5in]{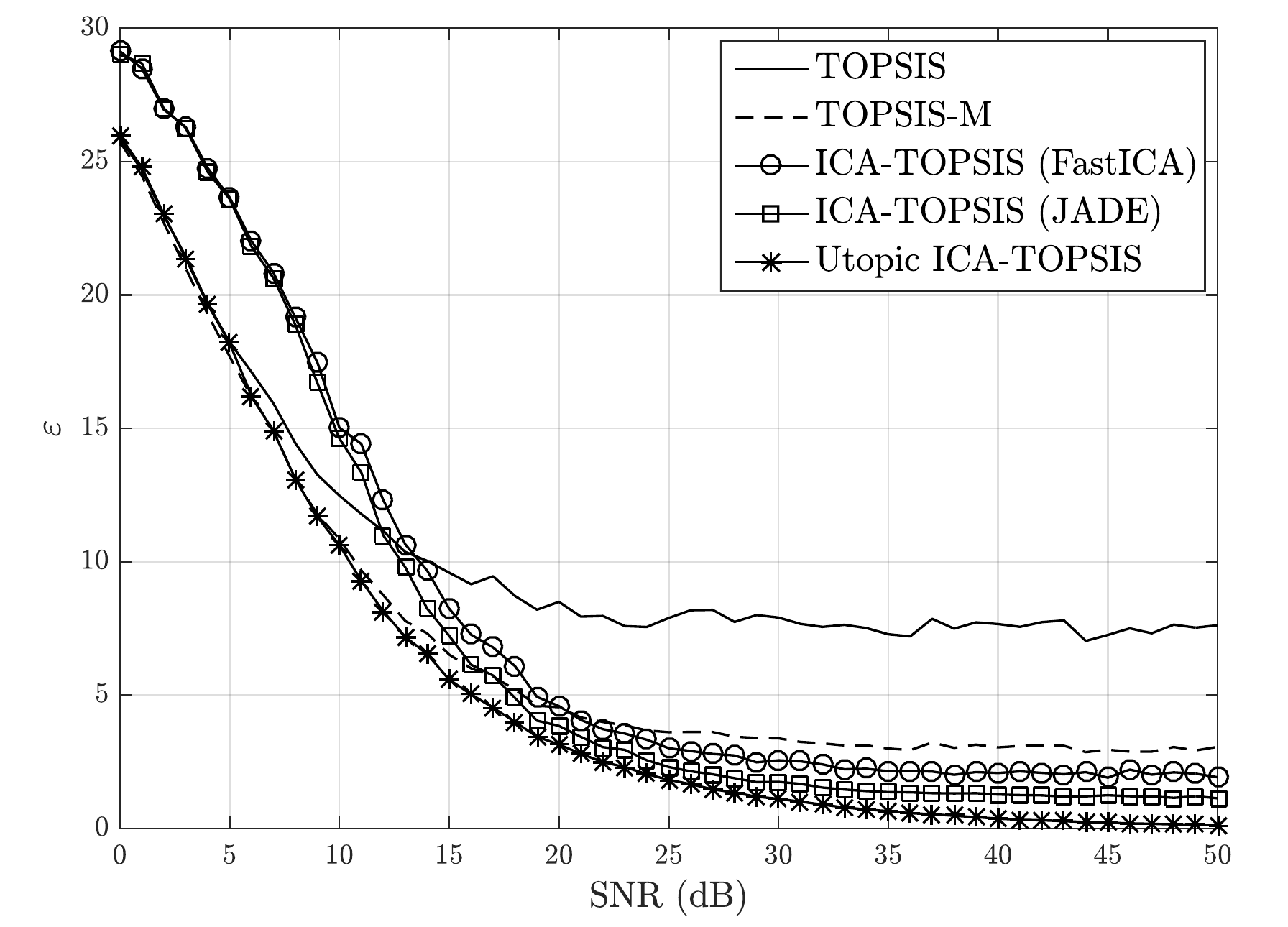}
\label{fig:snr2eit}}
\hfil
\subfloat[ICA-TOPSIS-M approach]{\includegraphics[width=4.5in]{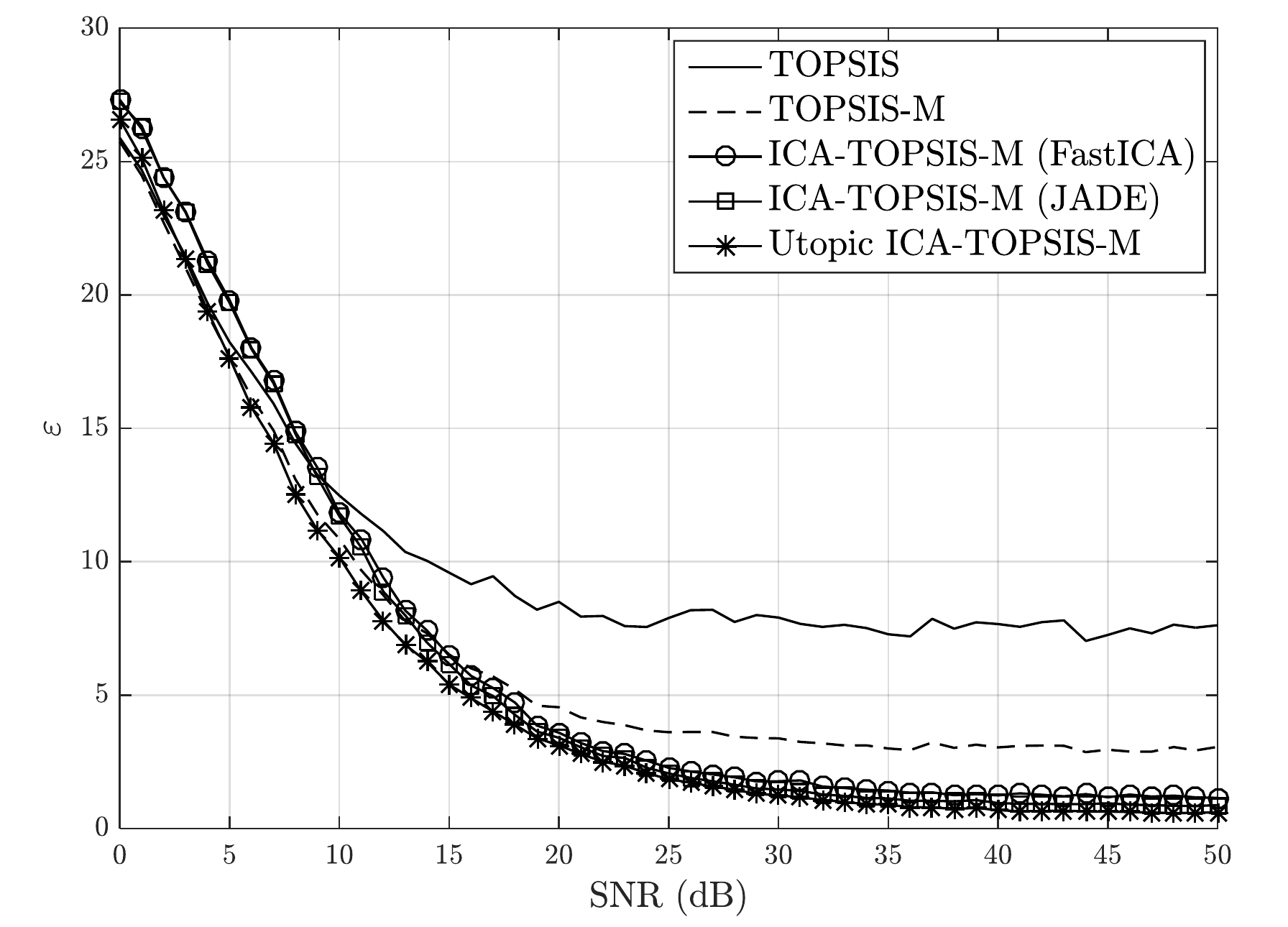}
\label{fig:snr2eitm}}
\caption{Comparison of mean absolute errors for different SNR values.}
\label{fig:snr2e}
\end{figure}

With respect to the ICA algorithms, the methods based on JADE achieved the highest values of the considered performance indexes. This is directly related to the performance in estimating the mixing matrix and, consequently, the latent criteria. Since JADE provided a better estimation of $\mathbf{A}$ compared to the FastICA, it also provided a ranking of alternatives closer to the target one.

\subsubsection{Comparison for different numbers of alternatives}
\label{subsubsec:exp3}

In the latter experiment we compared the considered approaches with a fix number of alternatives and by varying the level of noise in the mixing process. Conversely, in this third experiment, we compare them by fixing $SNR=30$ dB (a low noise level) and varying the number of alternatives. The results (averaged over 500 simulations for each number of alternatives) for the Kendall tau coefficient, Pearson's correlation coefficient and mean absolute error are illustrated in Figures~\ref{fig:nalt2k},~\ref{fig:nalt2c} and~\ref{fig:nalt2e}.

\begin{figure}[ht]
\centering
\subfloat[ICA-TOPSIS approach]{\includegraphics[width=4.5in]{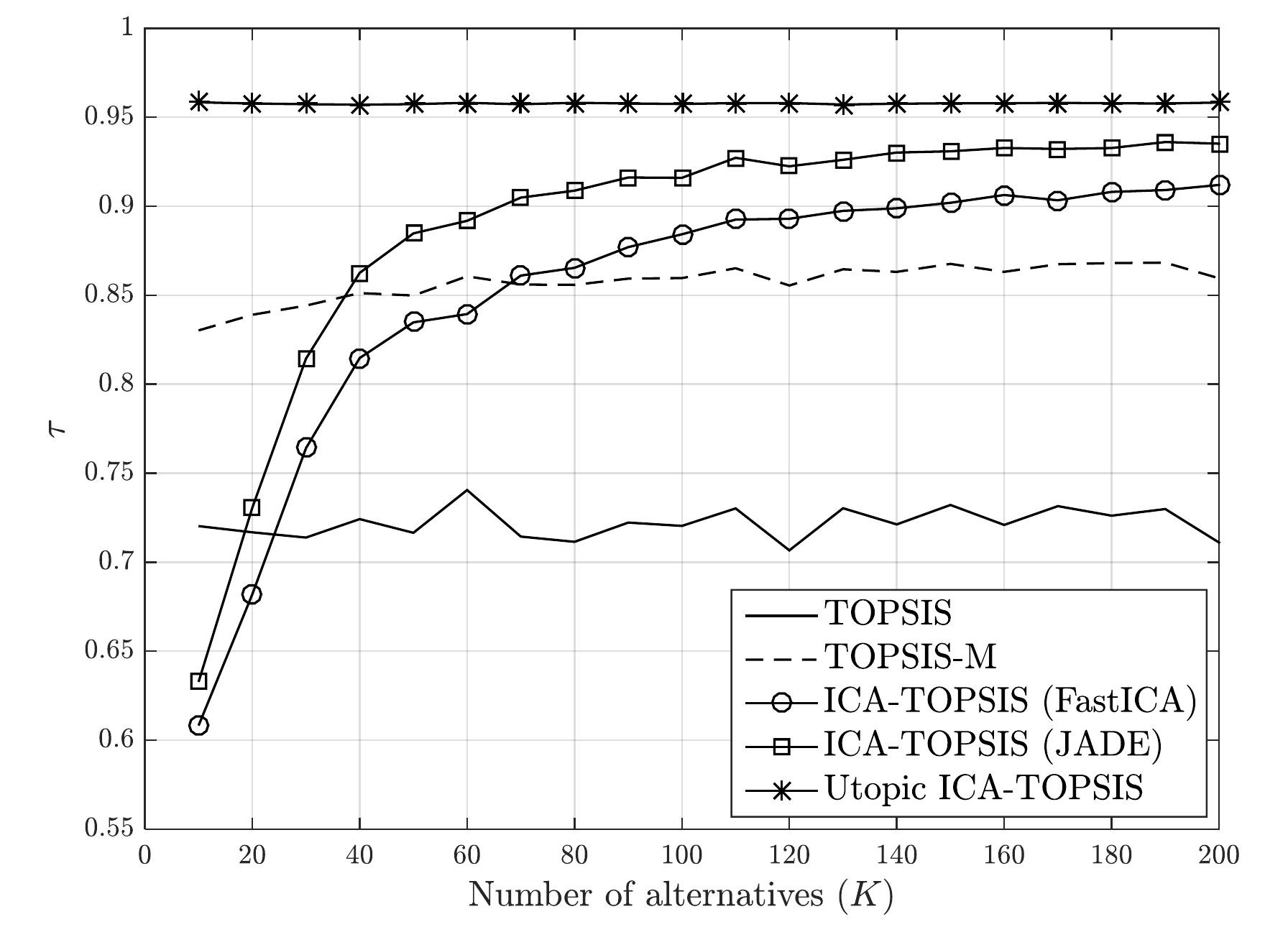}
\label{fig:nalt2kit}}
\hfil
\subfloat[ICA-TOPSIS-M approach]{\includegraphics[width=4.5in]{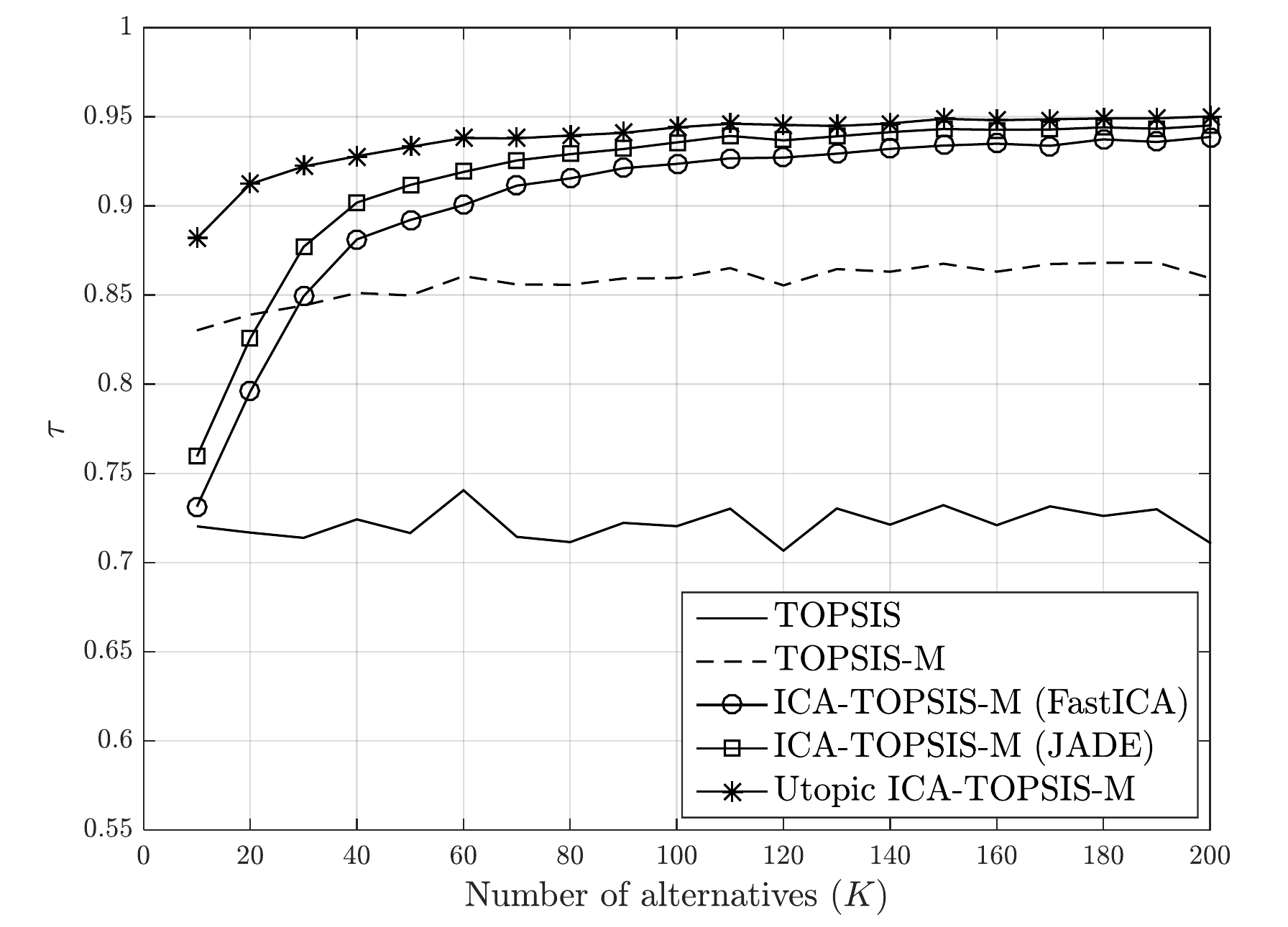}
\label{fig:nalt2kitm}}
\caption{Comparison of Kendall tau coefficients for different number of alternatives.}
\label{fig:nalt2k}
\end{figure}

\begin{figure}[ht]
\centering
\subfloat[ICA-TOPSIS approach]{\includegraphics[width=4.5in]{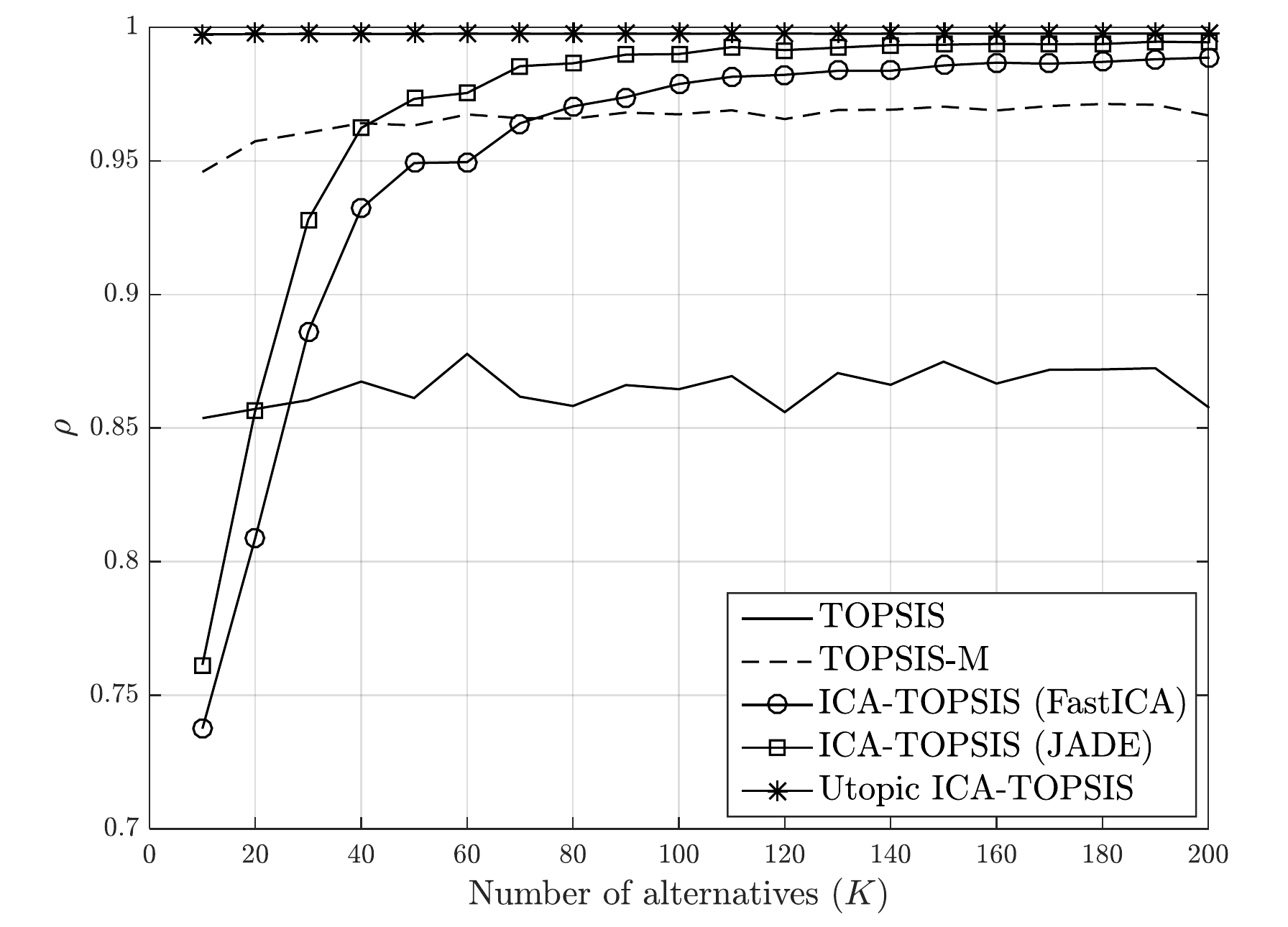}
\label{fig:nalt2cit}}
\hfil
\subfloat[ICA-TOPSIS-M approach]{\includegraphics[width=4.5in]{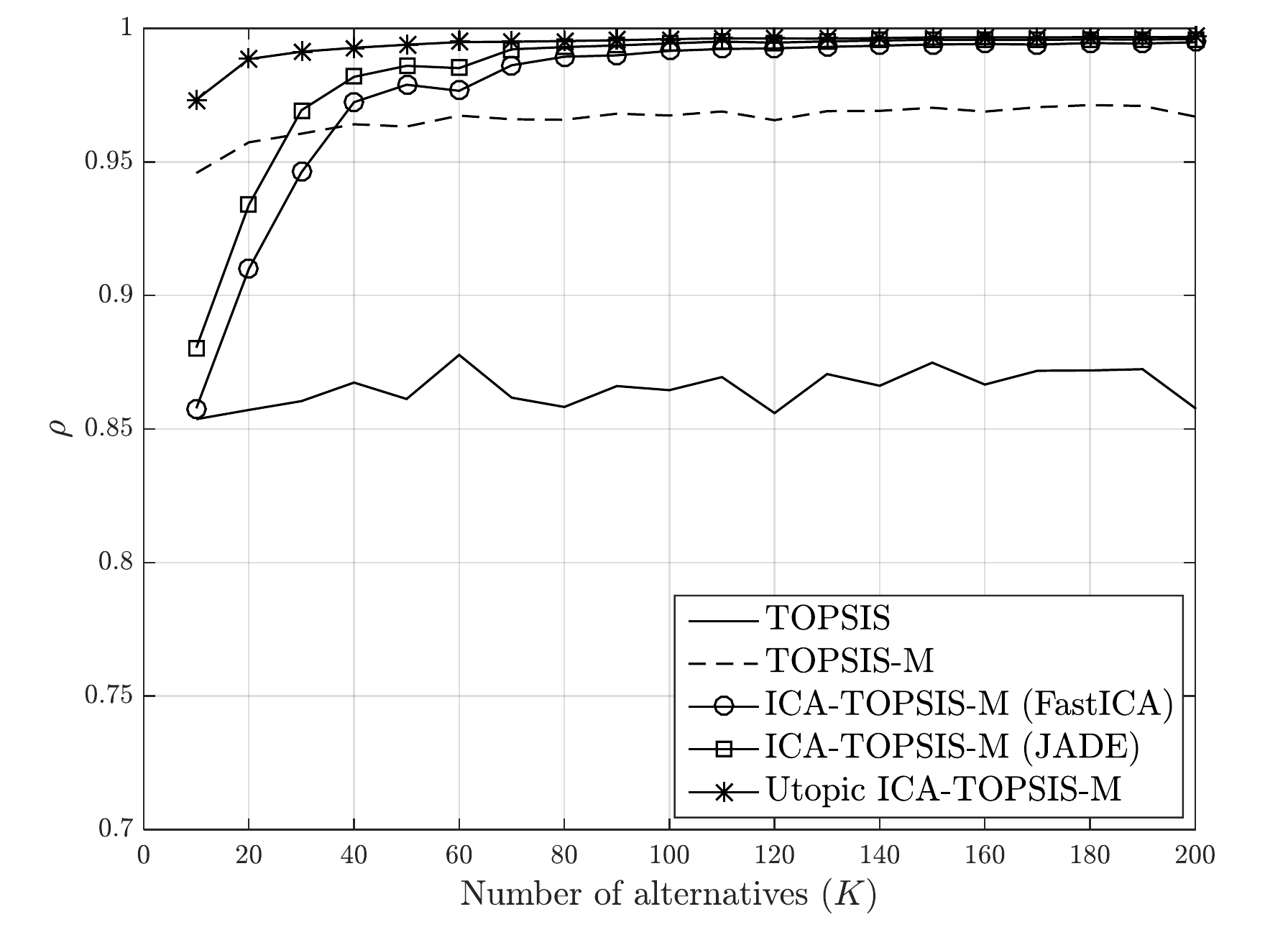}
\label{fig:nalt2citm}}
\caption{Comparison of Pearson's correlation coefficients for different number of alternatives.}
\label{fig:nalt2c}
\end{figure}

\begin{figure}[ht]
\centering
\subfloat[ICA-TOPSIS approach]{\includegraphics[width=4.5in]{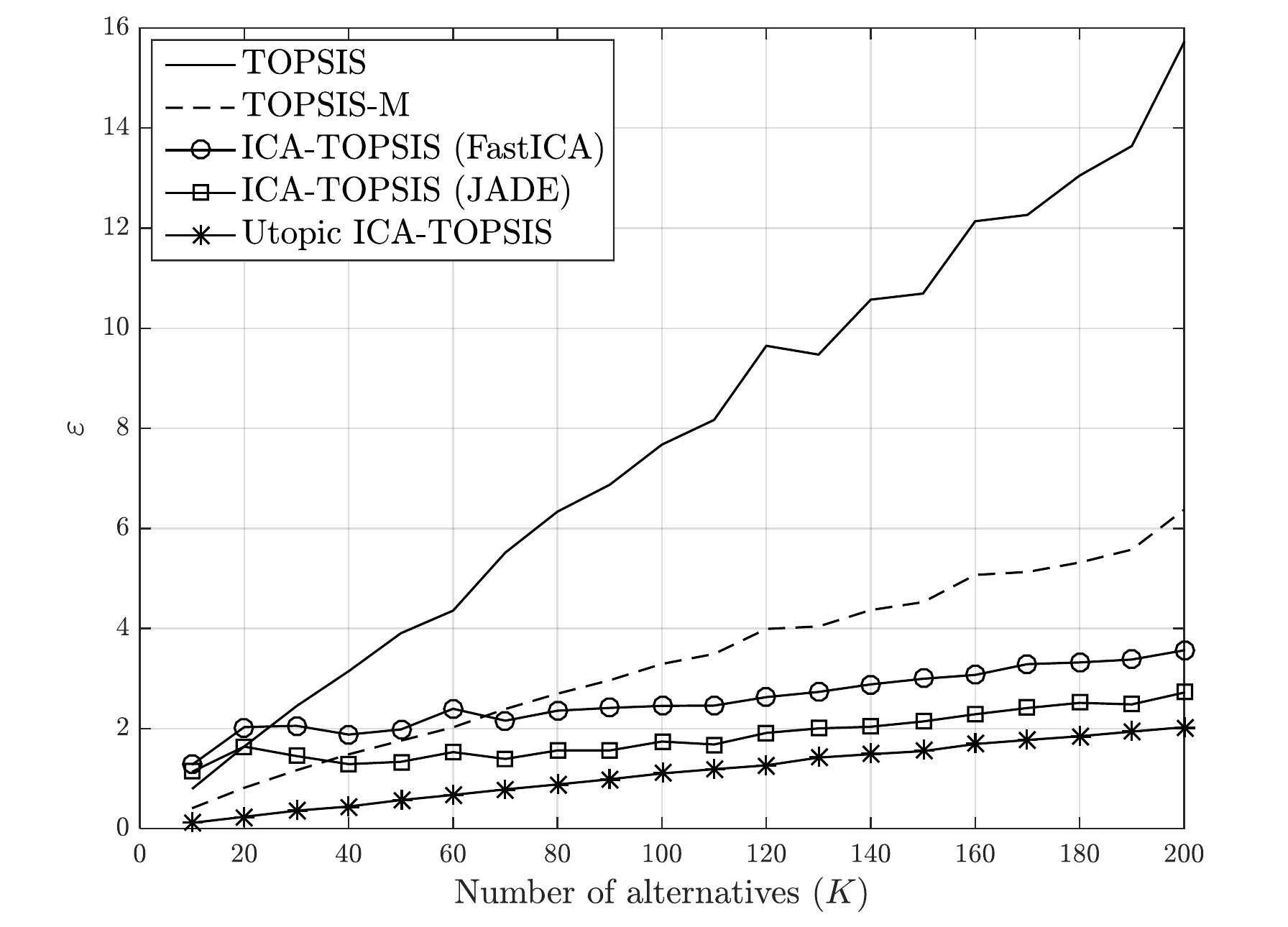}
\label{fig:nalt2eit}}
\hfil
\subfloat[ICA-TOPSIS-M approach]{\includegraphics[width=4.5in]{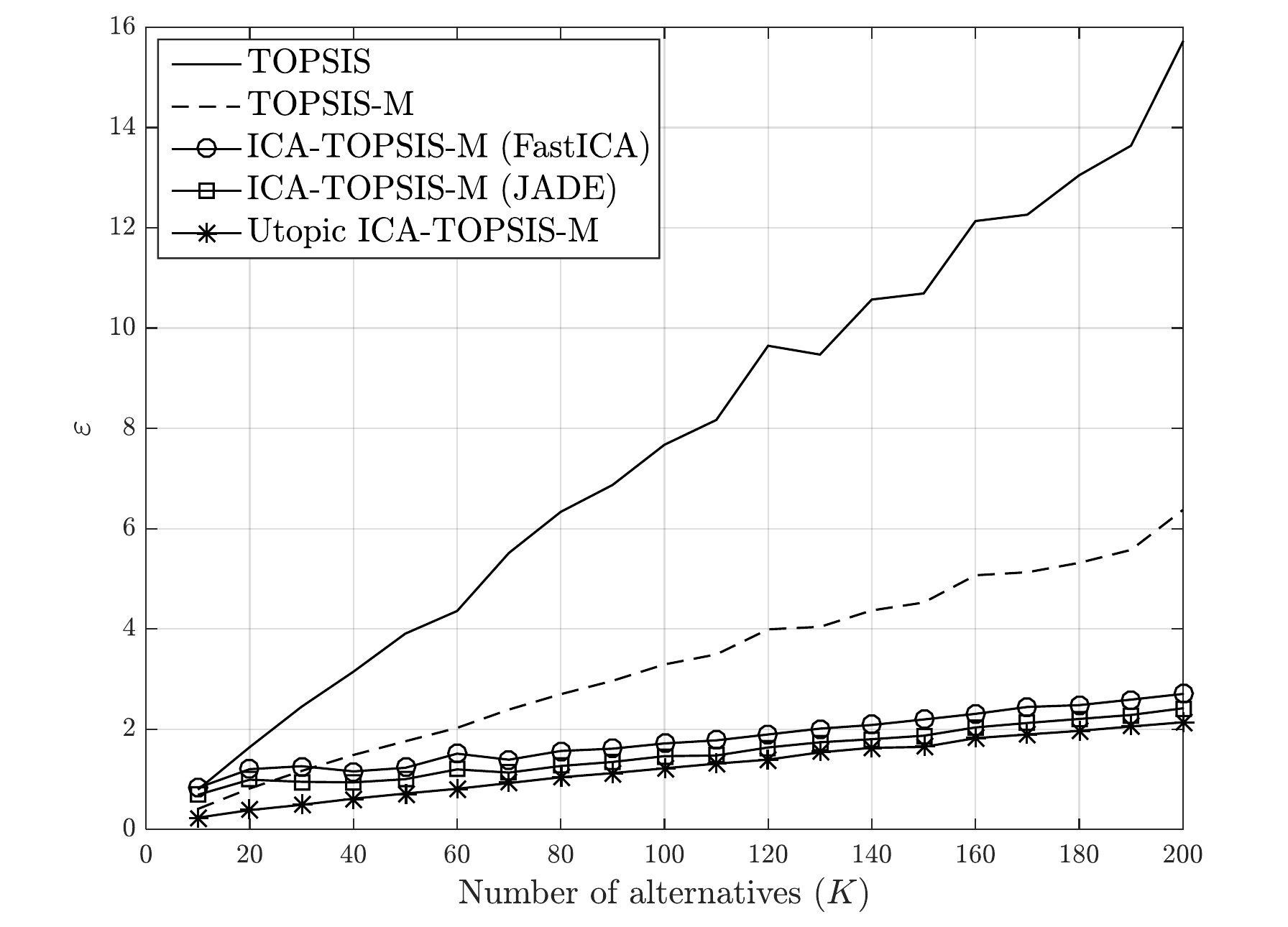}
\label{fig:nalt2eitm}}
\caption{Comparison of mean absolute error for different number of alternatives.}
\label{fig:nalt2e}
\end{figure}

Similarly as obtained in Section~\ref{subsubsec:exp2}, one may note that both ICA-TOPSIS and ICA-TOPSIS-M approaches performed better compared to TOPSIS and TOPSIS-M, specially for a number of alternatives greater 30. For $K$ lower than 30, the number of samples is not enough to the ICA technique perform a good estimation of the mixing matrix. If we compare the performance of the ICA algorithms, JADE also provided a better estimation of the mixing matrix and, consequently, the ranking of alternatives.

\subsubsection{Numerical experiments with more than two criteria}
\label{subsubsec:exp4}

In the aforementioned experiments, we considered a MCDM problem with only two criteria. All the obtained results indicated that the ICA-TOPSIS-M approach based on JADE algorithm can better deal with the addressed MCDM problem. In order to verify the robustness of this approach, we compare it with TOPSIS and TOPSIS-M in decision problems with more than two criteria. For instance, we consider scenarios with different levels of noise ($SNR=15$, $SNR=30$ and $SNR=45$ dB) and different numbers of alternatives ($K=30$, $K=100$ and $K=170$). The average value and standard deviation $\sigma$ (over 1000 simulations) of each performance index for 3, 4 and 5 decision criteria are shown in Tables~\ref{tab:3crit},~\ref{tab:4crit} and~\ref{tab:5crit}, respectively. The best result for each scenario and each performance index is indicated in bold.

\begin{table}[h!]
\centering
\tiny
\caption{Results for a decision problem with $M=3$ criteria.}
\setlength{\tabcolsep}{3pt}
\renewcommand{\arraystretch}{1.5}
\begin{tabular}{|c|c|c|c|c|c|c|c|c|c|c|}
\hline
\textbf{Performance} & \multirow{ 2}{*}{\textbf{Method}} & \multicolumn{9}{c|}{\textbf{Scenarios:} $\left\{SNR,K\right\}$} \\
\cline{3-11}
\textbf{index} &  & $\left\{15,30 \right\}$ & $\left\{15,100 \right\}$ & $\left\{15,170 \right\}$ & $\left\{30,30 \right\}$ & $\left\{30,100 \right\}$ & $\left\{30,170 \right\}$ & $\left\{45,30 \right\}$ & $\left\{45,100 \right\}$ & $\left\{45,170 \right\}$ \\
\hline
 & \multirow{ 2}{*}{TOPSIS} & 0.5088 & 0.5087 & 0.5043 & 0.5429 & 0.5365 & 0.5434 & 0.5372 & 0.5389 & 0.5228 \\
 &  & (0.1784) & (0.1671) & (0.1665) & (0.2026) & (0.1947) & (0.1932) & (0.2061) & (0.2052) & (0.1945) \\
\cline{2-11}
$\tau$ & \multirow{ 2}{*}{TOPSIS-M} & \textbf{0.6402} & 0.6561 & 0.6585 & 0.7235 & 0.7437 & 0.7458 & 0.7332 & 0.7455 & 0.7472 \\
$(\sigma_{\tau})$ &  & (0.1021) & (0.0761) & (0.0742) & (0.1299) & (0.1179) & (0.1133) & (0.1308) & (0.1239) & (0.1201) \\
\cline{2-11}
 & ICA-TOPSIS-M & 0.5667 & \textbf{0.6892} & \textbf{0.7167} & \textbf{0.7448} & \textbf{0.9089} & \textbf{0.9258} & \textbf{0.7808} & \textbf{0.9277} & \textbf{0.9478} \\
 & (JADE) & (0.1979) & (0.0983) & (0.0691) & (0.2275) & (0.0337) & (0.0235) & (0.1842) & (0.0400) & (0.0269) \\
\hline
 & \multirow{ 2}{*}{TOPSIS} & 0.6900 & 0.6940 & 0.6901 & 0.7170 & 0.7149 & 0.7207 & 0.7105 & 0.7114 & 0.6985 \\
 &  & (0.1929) & (0.1817) & (0.1821) & (0.2036) & (0.1993) & (0.1959) & (0.2041) & (0.2070) & (0.1973) \\
\cline{2-11}
$\rho$ & \multirow{ 2}{*}{TOPSIS-M} & \textbf{0.8358} & 0.8495 & 0.8519 & \textbf{0.8933} & 0.9055 & 0.9075 & 0.8971 & 0.9048 & 0.9064 \\
$(\sigma_{\rho})$ &  & (0.0822) & (0.0629) & (0.0619) & (0.0914) & (0.0795) & (0.0767) & (0.0871) & (0.0822) & (0.0784) \\
\cline{2-11}
 & ICA-TOPSIS-M & 0.7456 & \textbf{0.8710} & \textbf{0.8957} & 0.8787 & \textbf{0.9893} & \textbf{0.9931} & \textbf{0.9092} & \textbf{0.9926} & \textbf{0.9963} \\
 & (JADE) & (0.2300) & (0.0989) & (0.0662) & (0.2347) & (0.0086) & (0.0047) & (0.1714) & (0.0139) & (0.0037) \\
\hline
 & \multirow{ 2}{*}{TOPSIS} & 0.9157 & 3.0093 & 5.1098 & 0.8535 & 2.8094 & 4.6992 & 0.8838 & 2.8678 & 5.0050 \\
 &  & (0.5336) & (1.4848) & (2.4835) & (0.5579) & (1.6000) & (2.7342) & (0.5780) & (1.7304) & (2.7579) \\
\cline{2-11}
$\varepsilon$ & \multirow{ 2}{*}{TOPSIS-M} & \textbf{0.5943} & 1.8161 & 3.0007 & \textbf{0.4359} & 1.2642 & 2.1041 & 0.4205 & 1.2642 & 2.1007 \\
$(\sigma_{\varepsilon})$ &  & (0.3057) & (0.5962) & (0.9661) & (0.3138) & (0.7348) & (1.1600) & (0.2918) & (0.7806) & (1.2339) \\
\cline{2-11}
 & ICA-TOPSIS-M & 0.7757 & \textbf{1.6526} & \textbf{2.4284} & 0.4428 & \textbf{0.4003} & \textbf{0.5303} & \textbf{0.3749} & \textbf{0.2999} & \textbf{0.3584} \\
 & (JADE) & (0.5646) & (0.8438) & (0.9990) & (0.5922) & (0.1910) & (0.1978) & (0.4669) & (0.2100) & (0.2016) \\
\hline
\end{tabular}
\label{tab:3crit} 
\end{table}

\begin{table}[h!]
\centering
\tiny
\caption{Results for a decision problem with $M=4$ criteria.}
\setlength{\tabcolsep}{3pt}
\renewcommand{\arraystretch}{1.5}
\begin{tabular}{|c|c|c|c|c|c|c|c|c|c|c|}
\hline
\textbf{Performance} & \multirow{ 2}{*}{\textbf{Method}} & \multicolumn{9}{c|}{\textbf{Scenarios:} $\left\{SNR,K\right\}$} \\
\cline{3-11}
\textbf{index} &  & $\left\{15,30 \right\}$ & $\left\{15,100 \right\}$ & $\left\{15,170 \right\}$ & $\left\{30,30 \right\}$ & $\left\{30,100 \right\}$ & $\left\{30,170 \right\}$ & $\left\{45,30 \right\}$ & $\left\{45,100 \right\}$ & $\left\{45,170 \right\}$ \\
\hline
 & \multirow{ 2}{*}{TOPSIS} & 0.4257 & 0.4275 & 0.4238 & 0.4466 & 0.4427 & 0.4301 & 0.4409 & 0.4408 & 0.4319 \\
 &  & (0.1854) & (0.1690) & (0.1618) & (0.2011) & (0.1916) & (0.1913) & (0.2074) & (0.1884) & (0.1906) \\
\cline{2-11}
$\tau$ & \multirow{ 2}{*}{TOPSIS-M} & \textbf{0.5690} & 0.5829 & 0.5867 & \textbf{0.6353} & 0.6520 & 0.6478 & \textbf{0.6253} & 0.6489 & 0.6455 \\
$(\sigma_{\tau})$ &  & (0.1165) & (0.1000) & (0.0934) & (0.1435) & (0.1364) & (0.1542) & (0.1592) & (0.1378) & (0.1519) \\
\cline{2-11}
 & ICA-TOPSIS-M & 0.4410 & \textbf{0.5969} & \textbf{0.6488} & 0.5961 & \textbf{0.8793} & \textbf{0.9066} & 0.5971 & \textbf{0.9049} & \textbf{0.9317} \\
 & (JADE) & (0.2060) & (0.1644) & (0.1247) & (0.2607) & (0.0740) & (0.0433) & (0.2597) & (0.0434) & (0.0282) \\
\hline
 & \multirow{ 2}{*}{TOPSIS} & 0.5984 & 0.6032 & 0.6003 & 0.6178 & 0.6148 & 0.6016 & 0.6081 & 0.6139 & 0.6038 \\
 &  & (0.2210) & (0.2039) & (0.1957) & (0.2306) & (0.2199) & (0.2244) & (0.2386) & (0.2147) & (0.2224) \\
\cline{2-11}
$\rho$ & \multirow{ 2}{*}{TOPSIS-M} & \textbf{0.7709} & \textbf{0.7825} & 0.7863 & \textbf{0.8245} & 0.8373 & 0.8314 & \textbf{0.8118} & 0.8338 & 0.8280 \\
$(\sigma_{\rho})$ &  & (0.1130) & (0.0982) & (0.0919) & (0.1251) & (0.1179) & (0.1507) & (0.1470) & (0.1178) & (0.1475) \\
\cline{2-11}
 & ICA-TOPSIS-M & 0.6082 & 0.7824 & \textbf{0.8356} & 0.7515 & \textbf{0.9774} & \textbf{0.9879} & 0.7504 & \textbf{0.9879} & \textbf{0.9940} \\
 & (JADE) & (0.2609) & (0.1931) & (0.1374) & (0.2916) & (0.0632) & (0.0368) & (0.2844) & (0.0172) & (0.0048) \\
\hline
 & \multirow{ 2}{*}{TOPSIS} & 1.1511 & 3.7380 & 6.4295 & 1.1084 & 3.6643 & 6.3784 & 1.1105 & 3.6518 & 6.4229 \\
 &  & (0.6152) & (1.6822) & (2.6437) & (0.6111) & (1.7920) & (3.0382) & (0.6289) & (1.7573) & (3.0785) \\
\cline{2-11}
$\varepsilon$ & \multirow{ 2}{*}{TOPSIS-M} & \textbf{0.7628} & 2.3556 & 3.9702 & \textbf{0.6288} & 1.8876 & 3.2587 & \textbf{0.6345} & 1.8933 & 3.2917 \\
$(\sigma_{\varepsilon})$ &  & (0.3851) & (0.8902) & (1.3315) & (0.3859) & (1.0053) & (2.0816) & (0.4180) & (0.9950) & (2.0601) \\
\cline{2-11}
 & ICA-TOPSIS-M & 1.1075 & \textbf{2.3554} & \textbf{3.2790} & 0.7673 & \textbf{0.5652} & \textbf{0.7069} & 0.7549 & \textbf{0.4227} & \textbf{0.4938} \\
 & (JADE) & (0.6696) & (1.5372) & (1.8270) & (0.7120) & (0.5634) & (0.5697) & (0.7086) & (0.2513) & (0.2239) \\
\hline
\end{tabular}
\label{tab:4crit} 
\end{table}

\begin{table}[h!]
\centering
\tiny
\caption{Results for a decision problem with $M=5$ criteria.}
\setlength{\tabcolsep}{3pt}
\renewcommand{\arraystretch}{1.5}
\begin{tabular}{|c|c|c|c|c|c|c|c|c|c|c|}
\hline
\textbf{Performance} & \multirow{ 2}{*}{\textbf{Method}} & \multicolumn{9}{c|}{\textbf{Scenarios:} $\left\{SNR,K\right\}$} \\
\cline{3-11}
\textbf{index} &  & $\left\{15,30 \right\}$ & $\left\{15,100 \right\}$ & $\left\{15,170 \right\}$ & $\left\{30,30 \right\}$ & $\left\{30,100 \right\}$ & $\left\{30,170 \right\}$ & $\left\{45,30 \right\}$ & $\left\{45,100 \right\}$ & $\left\{45,170 \right\}$ \\
\hline
 & \multirow{ 2}{*}{TOPSIS} & 0.3659 & 0.3623 & 0.3606 & 0.3710 & 0.3671 & 0.3720 & 0.3661 & 0.3726 & 0.3682 \\
 &  & (0.1848) & (0.1529) & (0.1584) & (0.1895) & (0.1738) & (0.1656) & (0.1932) & (0.1755) & (0.1661) \\
\cline{2-11}
$\tau$ & \multirow{ 2}{*}{TOPSIS-M} & \textbf{0.4948} & \textbf{0.5238} & 0.5206 & \textbf{0.5368} & 0.5586 & 0.5714 & \textbf{0.5293} & 0.5647 & 0.5662 \\
$(\sigma_{\tau})$ &  & (0.1322) & (0.1087) & (0.1164) & (0.1671) & (0.1623) & (0.1589) & (0.1786) & (0.1776) & (0.1817) \\
\cline{2-11}
 & ICA-TOPSIS-M & 0.3383 & 0.4883 & \textbf{0.5567} & 0.4322 & \textbf{0.8302} & \textbf{0.8802} & 0.4215 & \textbf{0.8662} & \textbf{0.9181} \\
 & (JADE) & (0.2113) & (0.1743) & (0.1725) & (0.2494) & (0.1280) & (0.0718) & (0.2627) & (0.0976) & (0.0376) \\
\hline
 & \multirow{ 2}{*}{TOPSIS} & 0.5272 & 0.5264 & 0.5248 & 0.5337 & 0.5285 & 0.5372 & 0.5247 & 0.5371 & 0.5328 \\
 &  & (0.2234) & (0.1941) & (0.2010) & (0.2329) & (0.2152) & (0.2037) & (0.2358) & (0.2172) & (0.2072) \\
\cline{2-11}
$\rho$ & \multirow{ 2}{*}{TOPSIS-M} & \textbf{0.6882} & \textbf{0.7217} & 0.7171 & \textbf{0.7291} & 0.7474 & 0.7599 & \textbf{0.7186} & 0.7506 & 0.7522 \\
$(\sigma_{\rho})$ &  & (0.1449) & (0.1198) & (0.1325) & (0.1752) & (0.1740) & (0.1714) & (0.1994) & (0.1993) & (0.2048) \\
\cline{2-11}
 & ICA-TOPSIS-M & 0.4805 & 0.6690 & \textbf{0.7411} & 0.5889 & \textbf{0.9488} & \textbf{0.9774} & 0.5719 & \textbf{0.9688} & \textbf{0.9909} \\
 & (JADE) & (0.2824) & (0.2136) & (0.2070) & (0.3073) & (0.1158) & (0.0573) & (0.3290) & (0.0777) & (0.0207) \\
\hline
 & \multirow{ 2}{*}{TOPSIS} & 1.3198 & 4.3508 & 7.3898 & 1.3007 & 4.3489 & 7.2475 & 1.3135 & 4.3116 & 7.3241 \\
 &  & (0.6289) & (1.5793) & (2.7033) & (0.6360) & (1.7759) & (2.8046) & (0.6422) & (1.7707) & (2.7918) \\
\cline{2-11}
$\varepsilon$ & \multirow{ 2}{*}{TOPSIS-M} & \textbf{0.9606} & \textbf{2.8516} & 4.8531 & \textbf{0.8413} & 2.6205 & 4.2270 & \textbf{0.8696} & 2.5992 & 4.3285 \\
$(\sigma_{\varepsilon})$ &  & (0.4579) & (1.0423) & (1.8217) & (0.4908) & (1.4413) & (2.3095) & (0.5239) & (1.5956) & (2.7424) \\
\cline{2-11}
 & ICA-TOPSIS-M & 1.4007 & 3.2389 & \textbf{4.5328} & 1.1393 & \textbf{0.8537} & \textbf{0.9379} & 1.1933 & \textbf{0.6471} & \textbf{0.6097} \\
 & (JADE) & (0.7291) & (1.7367) & (2.7336) & (0.7576) & (0.9756) & (0.8476) & (0.8199) & (0.6713) & (0.3811) \\
\hline
\end{tabular}
\label{tab:5crit} 
\end{table}

One may note that TOPSIS approach leaded to the worst result in all scenarios. Comparing TOPSIS-M with ICA-TOPSIS-M, the latter achieved the higher performance for all scenarios whose level of noise and number of alternatives were, at least, 30 and 100, respectively. Moreover, for all scenarios with 170 alternatives, including the one with $SNR=15$ dB, ICA-TOPSIS-M provided the better result. Most scenarios in which TOPSIS-M performed better in comparison with ICA-TOPSIS-M, the level of noise was 15 dB and the number of alternatives was 30. This is a consequence of the performance of the ICA algorithm, since the strong interference of noise and the few number of samples make the estimation process difficult.

\subsection{Experiment on real data}
\label{subsec:exp_real}

This experiment comprises the application of the proposed approaches to rank a set of $K=96$ countries based on their evaluation according to $M=3$ observed criteria: forest area (\% of land area), gross national income (GNI) per capita (current US\$) and life expectancy at birth (years). This data refers to 2015 and was collected from World Bank at http:\/\/www.worldbank.org\/. It is worth mentioning that we adopted this set of criteria only to illustrate our proposal in a real data. Therefore, more criteria from different domains can be used in a problem of ranking countries.

Since the life expectancy at birth is affected by several factors, such as social and economic ones, this criterion may be dependent on the GNI. In order to verify this assumption, we present the scatter plots of each pair of criteria in Figures~\ref{fig:real_data12},~\ref{fig:real_data13} and~\ref{fig:real_data23}. One may remark that GNI per capita and life expectancy at birth have a certain degree of dependence. For instance, the correlation coefficient between these two criteria\footnote{For the others pairs of criteria, we have $\rho \left(\mathcal{C}_1, \mathcal{C}_2 \right) = 0.07$ and $\rho \left(\mathcal{C}_1, \mathcal{C}_3 \right) = 0.14$} is $\rho \left(\mathcal{C}_2, \mathcal{C}_3 \right) = 0.68$. Therefore, even when equal importance is attributed to the decision criteria, the application of TOPSIS on the observed data shall amplify the importance of the latent factors that drive both GNI per capita and life expectancy at birth. So, in this scenario, it becomes interesting to find a representation in which the latent criteria are independent, as exposed in Section~\ref{sec:prob}. These latent criteria may carry, for instance, information associated with environmental, economic and social aspects.

\begin{figure}[ht]
\centering
\subfloat[Scatter plot of $\mathcal{C}_1$ and $\mathcal{C}_2$.]{\includegraphics[width=2.6in]{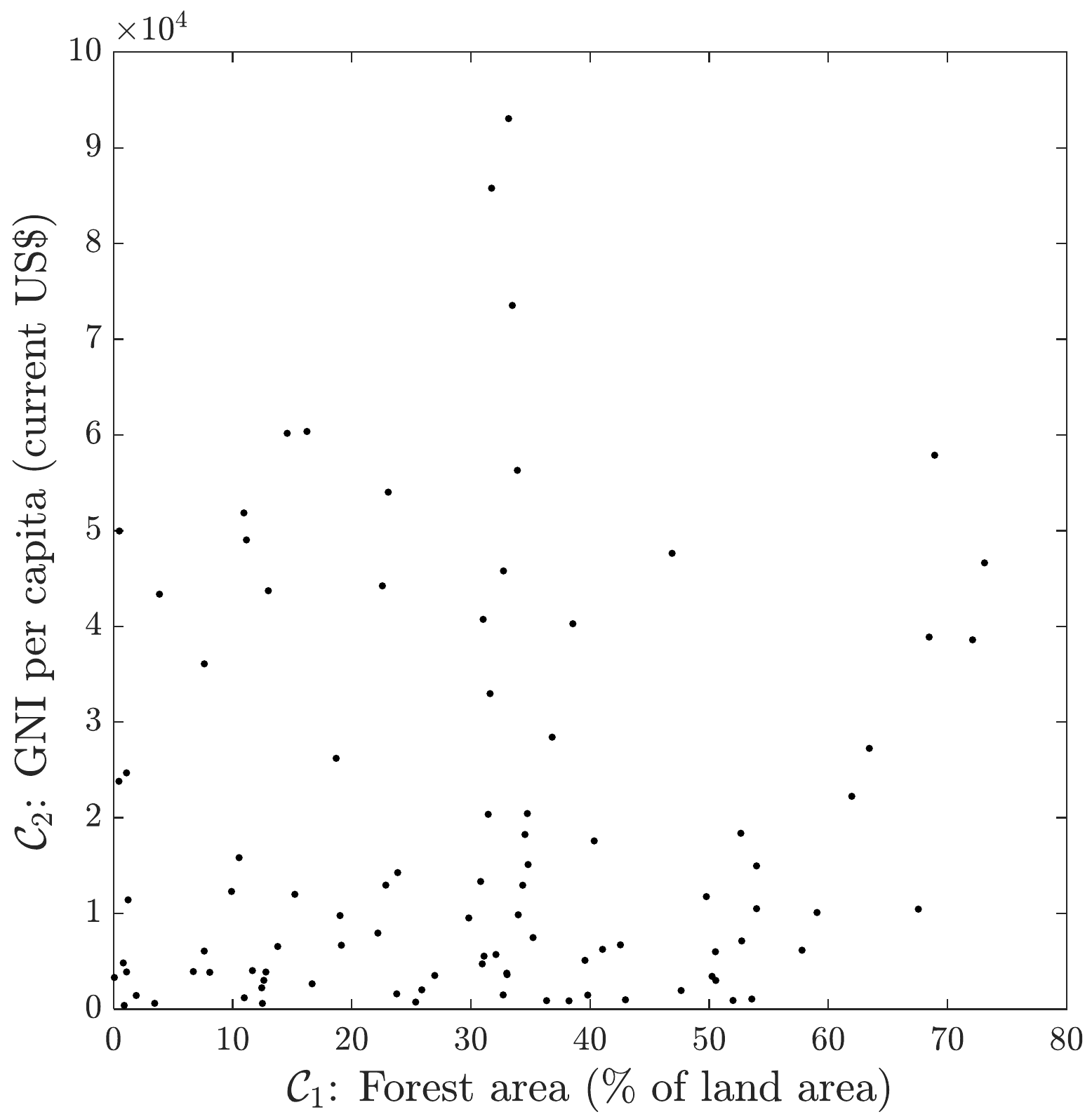}
\label{fig:real_data12}}
\hfil
\subfloat[Scatter plot of $\mathcal{C}_1$ and $\mathcal{C}_3$.]{\includegraphics[width=2.6in]{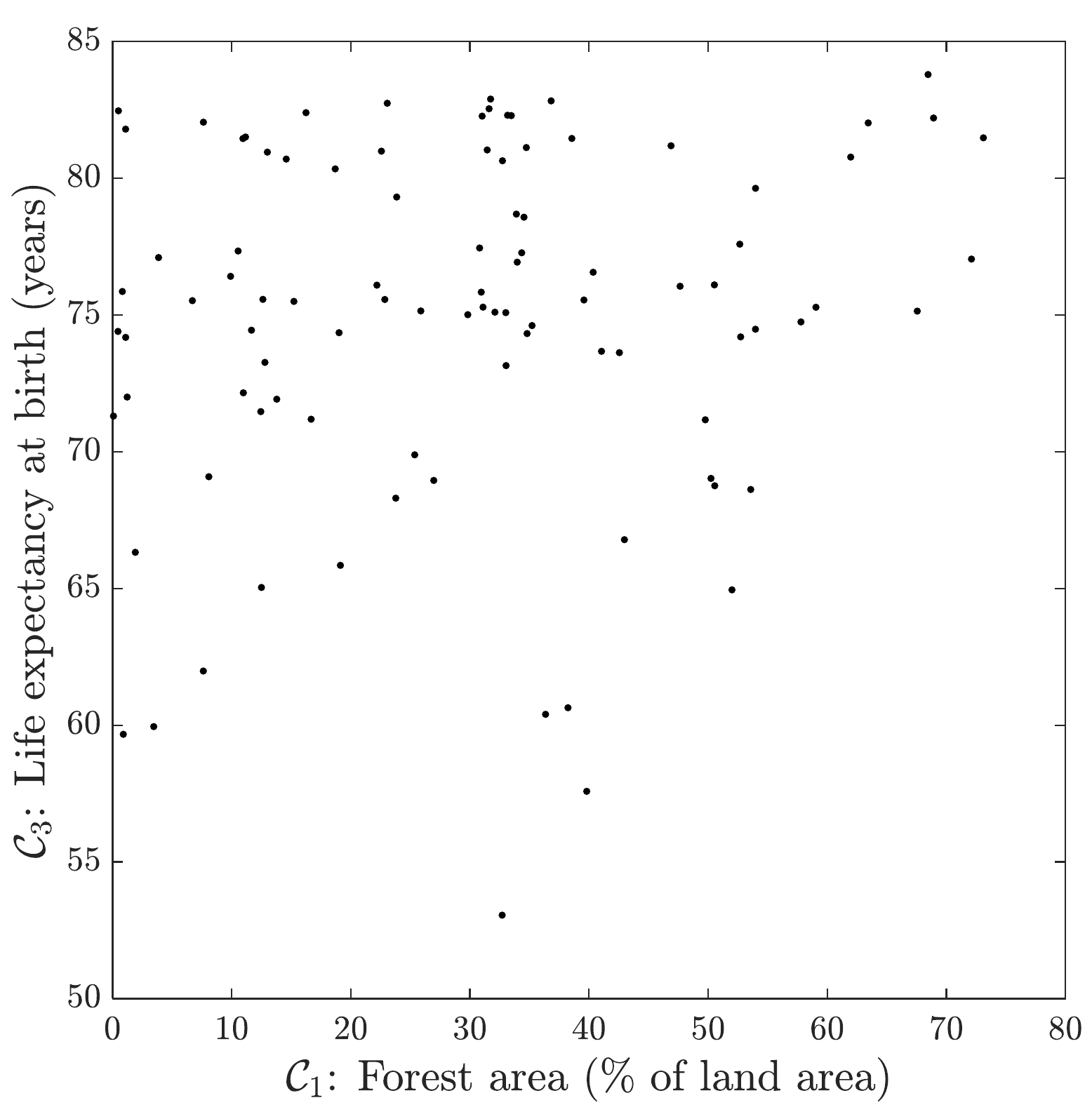}
\label{fig:real_data13}}
\hfil
\subfloat[Scatter plot of $\mathcal{C}_2$ and $\mathcal{C}_3$.]{\includegraphics[width=2.6in]{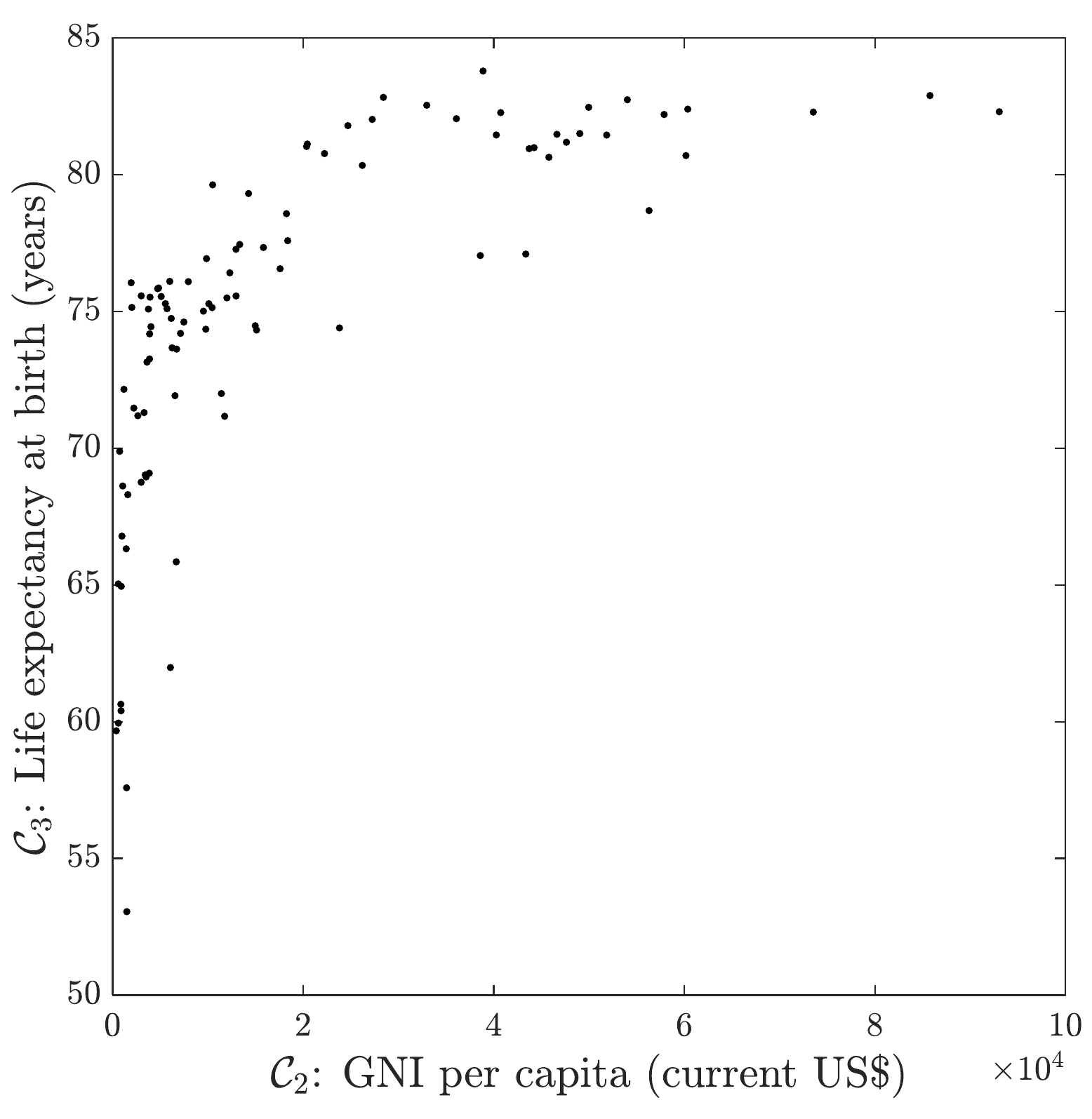}
\label{fig:real_data23}}
\caption{Real data visualization.}
\label{fig:real_data}
\end{figure}

Table~\ref{tab:comp_res} presents a comparison in the position of the first 10 alternatives in the ranking provided by each considered approaches. One may note in the ranking provided by TOPSIS that alternatives with good evaluations in both $\mathcal{C}_2$ and $\mathcal{C}_3$ are better ranked, even with a bad performance on the first criterion. For instance, the evaluations of the first alternative in the ranking, $\mathcal{A}_{66}$, are equal to $(33.1613; 93050; 82.3049)$.

\begin{table}[htb!]
  \begin{center}
	\footnotesize
  \caption{Rankings obtained in experiments with real data.}\label{tab:comp_res}
  {
  \setlength{\tabcolsep}{3pt}
  \renewcommand{\arraystretch}{1.5}
  \begin{tabular}{|c|ccc|cccc|}
    \hline
    \multirow{2}{*}{Alternatives} & \multicolumn{3}{|c|}{Criteria evaluations} & \multicolumn{4}{c|}{Position in the ranking} \\
     & $\mathcal{C}_1$ & $\mathcal{C}_2$ & $\mathcal{C}_3$ & TOPSIS & TOPSIS-M & ICA-TOPSIS & ICA-TOPSIS-M \\
    \hline
    $\mathcal{A}_{4}$ & 16.2388 & 60360 & 82.4000 & 10 & 15 & 59 & 15 \\
    \hline
		$\mathcal{A}_{5}$ & 46.8839 & 47630 & 82.4000 & 9 & 8 & 11 & 7 \\
    \hline
		$\mathcal{A}_{13}$ & 59.0488 & 10100 & 75.2840 & 28 & 24 & 8 & 30 \\
    \hline
		$\mathcal{A}_{14}$ & 72.1063 & 38590 & 77.0460 & 7 & 6 & 3 & 8 \\
    \hline
		$\mathcal{A}_{33}$ & 73.1072 & 46630 & 81.4805 & 5 & 4 & 2 & 2 \\
    \hline
		$\mathcal{A}_{46}$ & 68.4606 & 38880 & 83.7939 & 8 & 7 & 4 & 6 \\
    \hline
		$\mathcal{A}_{49}$ & 63.4386 & 27250 & 82.0244 & 14 & 10 & 5 & 9 \\
    \hline
		$\mathcal{A}_{50}$ & 53.9723 & 14970 & 74.4805 & 29 & 25 & 10 & 31 \\
    \hline
		$\mathcal{A}_{52}$ & 33.4749 & 73530 & 82.2927 & 4 & 5 & 25 & 5 \\
    \hline
		$\mathcal{A}_{54}$ & 67.5544 & 10450 & 75.1430 & 22 & 18 & 6 & 25 \\
    \hline
		$\mathcal{A}_{66}$ & 33.1613 & 93050 & 82.3049 & 1 & 2 & 21 & 4 \\
    \hline
		$\mathcal{A}_{68}$ & 57.7914 & 6160 & 74.7470 & 31 & 28 & 9 & 35 \\
    \hline
		$\mathcal{A}_{79}$ & 61.9662 & 22240 & 80.7756 & 18 & 12 & 7 & 14 \\
    \hline
		$\mathcal{A}_{83}$ & 68.9229 & 57880 & 82.2049 & 3 & 1 & 1 & 1 \\
    \hline
		$\mathcal{A}_{84}$ & 31.7340 & 85780 & 82.8976 & 2 & 3 & 23 & 3 \\
    \hline
		$\mathcal{A}_{93}$ & 33.8997 & 56300 & 78.6902 & 6 & 9 & 29 & 10 \\
		\hline
  \end{tabular}
  }
  \end{center}
\end{table}

On the other hand, the application of TOPSIS-M favours alternatives with a good evaluation on $\mathcal{C}_1$. In this case, the evaluations of the first alternative in the ranking, $\mathcal{A}_{83}$, are equal to $(68.9229; 57880; 82.2049)$. This is highlighted by the ranking provided by ICA-TOPSIS-M, in which the first two alternatives ($\mathcal{A}_{83}$ and $\mathcal{A}_{33}$) have good evaluation on $\mathcal{C}_1$.

An interesting result is obtained by the application of ICA-TOPSIS. Most of the first 10 alternatives in the ranking have good evaluations on the first criterion. Therefore, an alternative that performs well in $\mathcal{C}_1$ and in one of the other criteria is preferable compared to an alternative that performs well only in $\mathcal{C}_2$ and $\mathcal{C}_3$, which are dependent criteria. Moreover, some alternatives well evaluated in $\mathcal{C}_2$ and $\mathcal{C}_3$ that achieved good positions on the ranking provided by the other approaches were not well evaluated in ICA-TOPSIS. For instance, $\mathcal{A}_{66}$ was the first, second and forth alternative in the ranking provided by TOPSIS, TOPSIS-M and ICA-TOPSIS-M, respectively. However, by applying ICA-TOPSIS, this alternative achieved the position number 21.

\section{Conclusions}
\label{sec:concl}

Dependent criteria are frequently observed in real situations formulated as multicriteria decision making problems. Since most MCDM methods do not consider this relation among criteria on the aggregation procedure, the ranking of alternatives may be biased toward an alternative that has a good evaluation in the dependent criteria. With the goal of mitigating this biased effect, this paper proposed two different approaches, called ICA-TOPSIS and ICA-TOPSIS-M. In both approaches, the ICA technique can be seen as a preprocessing step whose aim is to extract information from the observed data in order to use as inputs of TOPSIS or TOPSIS-M.

By taking into account the experiments in synthetic data, in situations with two decision criteria, the obtained results indicates that both proposals performed better in comparison with TOPSIS and TOPSIS-M. Comparing ICA-TOPSIS and ICA-TOPSIS-M, the latter achieved the better results, specially with the application of JADE algorithm. The performance of ICA-TOPSIS-M based on JADE algorithm was also verified in scenarios with more than two decision criteria, which leads to the better results when we have moderate interference of noise and number of alternatives.

The application of the considered approaches on real data also provided interesting results. Since GNI and life expectancy at birth are dependent criteria, the use of TOPSIS directly on the observed decision data leads to a ranking in which the first alternatives are favoured by good evaluations in these two criteria. TOPSIS-M improved the results of TOPSIS, but both ICA-TOPSIS-M and ICA-TOPSIS leaded to rankings in which the importance of the other criteria is increased.

It is worth mentioning that, in order to apply the proposed approaches, one should consider the assumptions described in Section~\ref{subsec:icatopsis}. Therefore, a weakness of both ICA-TOPSIS and ICA-TOPSIS-M arises when these assumptions are not respected, since the ambiguities provided by the ICA technique may not be mitigated. In that respect, future works may be conducted to exploit other signal processing techniques that can deal with these ambiguities in the context of MCDM problems. Moreover, in this paper, we model and deal with linear dependence among criteria. Therefore, future perspectives also include the application of BSS methods that takes into account nonlinear relations among the decision data.

\vspace{2cm}

\noindent \textbf{Acknowledgements} \\

The authors would like to thank FAPESP (Processes n. 2016/21571-4 and 2017/23879-9) and CNPq (Processes n. 311357/2017-2) for the financial support. \\

%\noindent \textbf{References}
%\newpage

%\begin{thebibliography}{00}

\bibliographystyle{apacite}
\bibliography{_ref_eswa}

%\end{thebibliography}
\end{document}